\documentclass[onecolumn]{IEEEtran}

\usepackage[utf8]{inputenc}
\usepackage{amsmath}
\usepackage{amsthm}
\usepackage{mathtools}
\usepackage{amsfonts}
\usepackage{centernot}
\usepackage{hyperref}
\usepackage{url}
\usepackage{stackengine}
\usepackage{tikz}
\usepackage{pgfplots}
\pgfplotsset{compat=1.13}
\usepackage[nonumberlist,acronym]{glossaries}
\usepackage[noadjust]{cite}
\usepackage{aligned-overset}
\usepackage{todonotes}

\newcommand\blankfootnote[1]{%
  \begin{NoHyper}
  \let\svthefootnote\thefootnote%
  \let\thefootnote\relax\footnotetext{#1}%
  \let\thefootnote\svthefootnote%
  \end{NoHyper}
}

\newcommand{\bobPostproc}{{D}} 
\newcommand{\jammer}{\mathfrak{J}} 
\newcommand{\bob}{\mathfrak{B}} 
\newcommand{\eve}{\mathfrak{E}} 
\newcommand{\objectiveFunction}{{f}} 
\newcommand{\objectiveEstimator}{{\tilde{f}}} 
\newcommand{\inputRV}{{X}} 
\newcommand{\outputRV}{{Y}} 
\newcommand{\inputAlphabet}{{\mathcal{X}}} 
\newcommand{\inputAlphabetElement}{{x}} 
\newcommand{\outputAlphabet}{{\mathcal{Y}}} 
\newcommand{\outputAlphabetElement}{{y}} 
\newcommand{\stateSpace}{{\mathcal{S}}} 
\newcommand{\stateSpaceElement}{{s}} 
\newcommand{\compoundChannel}[1]{{W_{#1}}} 
\newcommand{\channel}{W} 
\newcommand{\approximateChannel}[2]{{\hat{W}_{{#1},{#2}}}} 
\newcommand{\messageRV}{\mathcal{M}} 
\newcommand{\information}[2]{{\mathbf{I}_{{#1},{#2}}}} 
\newcommand{\informationDensity}[4]{{\mathbf{i}_{{#1},{#2}}({#3};{#4})}} 
\newcommand{\codebook}{{\mathcal{C}}} 
\newcommand{\codeword}[1]{{\mathcal{C}(#1)}} 
\newcommand{\codewordsub}[2]{{\mathcal{C}_{#1}(#2)}} 
\newcommand{\codebookRate}{\mathcal{R}} 
\newcommand{\blocklength}{{n}} 
\newcommand{\inputDistribution}{{P}} 
\newcommand{\inputOutputDistribution}[2]{Q_{{#1},{#2}}} 
\newcommand{\marginalOutputDistribution}[2]{R_{{#1},{#2}}} 
\newcommand{\typicalityParameter}{{\varepsilon}} 
\newcommand{\Probability}{{\mathbb{P}}} 
\newcommand{\Expectation}{{\mathbb{E}}} 
\newcommand{\channelApproximationError}{{\delta}} 
\newcommand{\channelApproximationIndex}{{j}} 
\newcommand{\channelApproximationNumber}[1]{{J({#1})}} 
\newcommand{\generalChannelApproximationNumber}{{J}} 
\newcommand{\codewordIndex}{{m}} 
\newcommand{\errorevent}{{\mathcal{E}}} 
\newcommand{\boundParamOne}{{\beta_1}} 
\newcommand{\boundParamTwo}{{\beta_2}} 
\newcommand{\blockIndex}{{i}} 
\newcommand{\RNDerivative}[2]{\frac{d{#1}}{d{#2}}} 
\newcommand{\renyiParam}{{\alpha}} 
\newcommand{\renyiParamOne}{{\alpha_1}} 
\newcommand{\renyiParamTwo}{{\alpha_2}} 
\newcommand{\renyiParamThree}{{\alpha_3}} 
\newcommand{\renyidiv}[3]{\mathbf{D}_{#1}\left({#2} || {#3}\right)} 
\newcommand{\indicator}[1]{\mathbf{1}_{#1}} 
\newcommand{\kldiv}[2]{\renyidiv{1}{#1}{#2}} 
\newcommand{\absolutelyContinuous}{{\ll}} 
\newcommand{\notAbsolutelyContinuous}{{\centernot{\ll}}} 
\newcommand{\finalconst}{{\gamma}} 
\newcommand{\finalconstOne}{{\gamma_1}} 
\newcommand{\finalconstTwo}{{\gamma_2}} 
\newcommand{\finalconstThree}{{\gamma_3}} 
\newcommand{\finalconstFour}{{\gamma_4}} 
\newcommand{\analogMessageAlphabet}[1]{{\mathcal{S}_{#1}}} 
\newcommand{\analogMessageAlphabetElement}[1]{{s_{#1}}} 
\newcommand{\processedAnalogMessageRV}[1]{T_{#1}} 
\newcommand{\alice}[1]{{\mathfrak{A}_{#1}}} 
\newcommand{\aliceIndex}{{k}} 
\newcommand{\aliceNum}{{K}} 
\newcommand{\alicePreproc}[1]{F_{#1}} 
\newcommand{\eveAlphabet}{{\mathcal{Z}}} 
\newcommand{\eveOutputRV}{{Z}} 
\newcommand{\effectiveChannelBob}{W_{\mathfrak{B}}} 
\newcommand{\effectiveChannelEve}{W_{\mathfrak{E}}} 
\newcommand{\eveMarginalOutputDistribution}[1]{{\hat{R}_{#1}}} 
\newcommand{\totalvariation}[1]{\lVert #1 \rVert_\mathrm{TV}} 
\newcommand{\totalvariationlr}[1]{\left\lVert #1 \right\rVert_\mathrm{TV}} 
\newcommand{\noiseRV}{{N}} 
\newcommand{\absolute}[1]{\left| {#1} \right|} 
\newcommand{\reals}{\mathbb{R}} 
\newcommand{\coverElement}{{\mathcal{A}}} 
\newcommand{\numRxAntennas}{{i}} 
\newcommand{\numTxAntennas}{{j}} 
\newcommand{\channelMatrix}{{H}} 
\newcommand{\covarianceMatrix}{{\Sigma}} 
\newcommand{\psdSymMatrices}[1]{{\mathrm{Sym}_+^{#1}}} 
\newcommand{\pdSymMatrices}[1]{{\mathrm{Sym}_{++}^{#1}}} 
\newcommand{\mean}{{\mu}} 
\newcommand{\density}[1]{{p_{#1}}} 
\newcommand{\tail}{\varepsilon} 
\newcommand{\fadingAliceBob}{h_{\mathfrak{A}\mathfrak{B}}} 
\newcommand{\fadingAliceEve}{h_{\mathfrak{A}\mathfrak{E}}} 
\newcommand{\fadingJammerBob}{h_{\mathfrak{\mathfrak{J}\mathfrak{B}}}} 
\newcommand{\fadingJammerEve}{h_{\mathfrak{\mathfrak{J}\mathfrak{E}}}} 
\newcommand{\noiseEve}{N_{\mathfrak{E}}} 
\newcommand{\noiseBob}{N_{\mathfrak{B}}} 
\newcommand{\noiseEveComponent}[1]{N_{\mathfrak{E},{#1}}} 
\newcommand{\noiseBobComponent}[1]{N_{\mathfrak{B},{#1}}} 
\newcommand{\stddevBob}{\sigma_\mathfrak{B}} 
\newcommand{\stddevEve}{\sigma_\mathfrak{E}} 
\newcommand{\stddevBobEff}{\sigma_{\mathrm{eff}, \mathfrak{B}}} 
\newcommand{\stddevEveEff}{\sigma_{\mathrm{eff}, \mathfrak{E}}}  
\newcommand{\jammerPowerConstraint}{\mathfrak{P}_\mathfrak{J}} 
\newcommand{\alicePowerConstraint}{\mathfrak{P}_\mathfrak{A}} 
\newcommand{\alicePreprocElement}{t} 
\newcommand{\alicePreprocRandomness}{U} 
\newcommand{\alicePreprocInnerVariable}{u} 
\newcommand{\functionApproximationError}{\delta} 
\newcommand{\eveOutputDistributionGeneral}[1]{\tilde{R}_{#1}} 
\newcommand{\eveObjective}{g} 
\newcommand{\eveObjectiveRange}{\mathcal{T}} 
\newcommand{\eveDecoder}{d} 
\newcommand{\jammingDecoder}{\vartheta} 
\newcommand{\uselessDistribution}{\mu} 
\newcommand{\securityError}{\delta} 
\newcommand{\secLoss}{V} 
\newcommand{\mseError}{V} 
\newcommand{\intervalLowerBound}{a} 
\newcommand{\intervalUpperBound}{b} 
\newcommand{\mseFunction}{\Psi} 
\newcommand{\mseFunctionModified}{\hat{\Psi}} 
\newcommand{\stddev}{\sigma} 
\newcommand{\indicatorAbbrev}{\mathrm{ind}} 
\newcommand{\finiteApproxInteger}{M} 
\newcommand{\correspondingApproximateChannel}[1]{{W_{#1}'}} 
\newcommand{\cardinality}[1]{\lvert {#1} \rvert} 
\newcommand{\errorProb}{\epsilon} 

\newcommand{\generalrvOne}{U} 
\newcommand{\generalrvTwo}{V} 
\newcommand{\generalrvOneValue}{u} 
\newcommand{\generalrvTwoValue}{v} 

\newcommand{\generalSummationIndex}{k} 
\newcommand{\momentGeneratingFunction}{\varphi} 
\newcommand{\chernoffParam}{\lambda} 
\newcommand{\costConstraint}{C} 
\newcommand{\generalSummationBound}{n} 
\newcommand{\costFunction}{c} 
\newcommand{\lemmaBadCodewords}{\mathfrak{N}} 
\newcommand{\codebookBlocklength}{n} 
\newcommand{\proofconstantOne}{{\beta_1}} 
\newcommand{\proofconstantTwo}{{\beta_2}} 
\newcommand{\lemmapvalue}{p} 
\newcommand{\generalReal}{t} 
\newcommand{\generalRealTwo}{u} 
\newcommand{\generalRealThree}{v} 
\newcommand{\stdnormalcdf}{\Phi_N} 
\newcommand{\stdnormalpdf}{\varphi_N} 
\newcommand{\generalMeasureOne}{\mu} 
\newcommand{\generalMeasureTwo}{\nu} 
\newcommand{\genericEncoder}{e} 
\newcommand{\genericDecoder}{d} 
\newcommand{\generalSet}{S} 
\newcommand{\generalFunction}{f}

\newacronym{ota}{OTA}{over-the-air}
\newacronym{awgn}{AWGN}{additive white Gaussian noise}
\newacronym{dfa}{DFA}{distributed function approximation}
\newacronym{dfaj}{DFA-J}{distributed function approximation with jamming}
\newacronym{mse}{MSE}{mean square error}

\newtheorem{theorem}{Theorem}

\newtheorem{lemma}{Lemma}
\newtheorem{cor}{Corollary}

\newtheorem{remark}{Remark}
\newtheorem{definition}{Definition}

\definecolor{revcolor}{rgb}{0,0,0.5}

\makenoidxglossaries
\glsdisablehyper
\newglossaryentry{objectiveFunction}{
  name={\ensuremath{\objectiveFunction: \stateSpace_1 \times \ldots \times \stateSpace_\aliceNum\to\reals}},
  description={Objective function to be approximated}
}
\newglossaryentry{alicePreproc}{
  name={\ensuremath{\alicePreproc{}^\blocklength = (\alicePreproc{1}^\blocklength, \dots, \alicePreproc{\aliceNum}^\blocklength)}},
  description={Pre-processors for $\blocklength$ channel uses at the transmitters}
}
\newglossaryentry{channel}{
  name={\ensuremath{\channel}},
  description={channel}
}
\newglossaryentry{channelInAlph}{
  name={\ensuremath{\inputAlphabet_1, \dots, \inputAlphabet_\aliceNum}},
  description={multiple-access channel input alphabets}
}
\newglossaryentry{channelOutAlph}{
  name={\ensuremath{\outputAlphabet}},
  description={legitimate receiver's channel output alphabet; output alphabet of point-to-point channel}
}
\newglossaryentry{costConstraint}{
  name={\ensuremath{(\costFunction, \costConstraint)}},
  description={additive input cost constraint for a channel}
}
\newglossaryentry{bobPostproc}{
  name={\ensuremath{\bobPostproc^\blocklength}},
  description={Post-processor for $\blocklength$ channel uses at the receiver}
}
\newglossaryentry{objectiveEstimator}{
  name={\ensuremath{\objectiveEstimator}},
  description={Estimator at the receiver for $\objectiveFunction(\stateSpaceElement_1, \dots, \stateSpaceElement_\aliceNum)$}
}
\newglossaryentry{transmitters}{
  name={\ensuremath{\alice{1}, \dots, \alice{\aliceNum}}},
  description={transmitters}
}
\newglossaryentry{receiver}{
  name={\ensuremath{\bob}},
  description={legitimate receiver}
}
\newglossaryentry{jammer}{
  name={\ensuremath{\jammer}},
  description={jammer}
}
\newglossaryentry{eavesdropper}{
  name={\ensuremath{\eve}},
  description={eavesdropper}
}
\newglossaryentry{compoundChannel}{
  name={\ensuremath{(\compoundChannel{\stateSpaceElement})_{\stateSpaceElement \in \stateSpace}}},
  description={compound channel}
}
\newglossaryentry{inAlph}{
  name={\ensuremath{\inputAlphabet}},
  description={jammer's input alphabet; input alphabet of point-to-point channel}
}
\newglossaryentry{renyidiv}{
  name={\ensuremath{\renyidiv{\renyiParam}{\cdot}{\cdot}}},
  text={\ensuremath{\renyidiv{\renyiParam}{\generalMeasureOne}{\generalMeasureTwo}}},
  description={Rényi divergence of order $\renyiParam$}
}
\newglossaryentry{kldiv}{
  name={\ensuremath{\kldiv{\cdot}{\cdot}}},
  text={\ensuremath{\renyidiv{1}{\generalMeasureOne}{\generalMeasureTwo}}},
  description={Kullback-Leibler divergence}
}
\newglossaryentry{mutInf}{
  name={\ensuremath{\information{\inputDistribution}{\channel}}},
  description={Mutual information between input and output of channel $\channel$ under input distribution $\inputDistribution$}
}
\newglossaryentry{codebookensemble}{
  name={$(\inputDistribution, \blocklength, \codebookRate)$-ensemble},
  description={random codebook ensemble with input distribution $\inputDistribution$, block length $\blocklength$ and rate $\codebookRate$}
}
\newglossaryentry{effectiveChannelBob}{
  name={\ensuremath{\effectiveChannelBob}},
  description={legitimate user's effective channel}
}
\newglossaryentry{effectiveChannelEve}{
  name={\ensuremath{\effectiveChannelEve}},
  description={eavesdropper's effective channel}
}
\newglossaryentry{eveAlph}{
  name={\ensuremath{\eveAlphabet}},
  description={eavesdropper's channel output alphabet}
}
\newglossaryentry{jointDist}{
  name={\ensuremath{\inputOutputDistribution{\inputDistribution}{\channel}}},
  description={joint input-output distribution of channel $\channel$ under input distribution $\inputDistribution$}
}
\newglossaryentry{marginalDist}{
  name={\ensuremath{\marginalOutputDistribution{\inputDistribution}{\channel}}},
  description={output distribution of channel $\channel$ under input distribution $\inputDistribution$}
}
\newglossaryentry{psdsym}{
  name={\ensuremath{\psdSymMatrices{\generalSummationBound}}},
  description={symmetric, positive semidefinite $\generalSummationBound \times \generalSummationBound$ matrices}
}
\newglossaryentry{pdsym}{
  name={\ensuremath{\pdSymMatrices{\generalSummationBound}}},
  description={symmetric, positive definite $\generalSummationBound \times \generalSummationBound$ matrices}
}

\newglossaryentry{processedAnalogMessageRV}{
  text={\ensuremath{\processedAnalogMessageRV{\aliceIndex}}},
  name={\ensuremath{\processedAnalogMessageRV{1}, \dots, \processedAnalogMessageRV{\aliceNum}}},
  description={transmitters' pre-processed channel inputs}
}
\newglossaryentry{inputRV}{
  name={\ensuremath{\inputRV}},
  description={jammer's channel input; input of point-to-point channel}
}
\newglossaryentry{outputRV}{
  name={\ensuremath{\outputRV}},
  description={legitimate receiver's channel output; output of point-to-point channel}
}
\newglossaryentry{approximateChannel}{
  name={\ensuremath{(\approximateChannel{\channelApproximationError}{\channelApproximationIndex})_{\channelApproximationIndex=1}^{\generalChannelApproximationNumber}}},
  description={sequence of channels that approximate a compound channel with error $\channelApproximationError$}
}
\newglossaryentry{messageRV}{
  name={\ensuremath{\messageRV}},
  description={randomness used in jamming strategy; transmitted message in compound channel}
}
\newglossaryentry{informationDensity}{
  name={\ensuremath{\informationDensity{\inputDistribution}{\channel}{\inputAlphabetElement^\blocklength}{\outputAlphabetElement^\blocklength}}},
  description={information density of the input-output pair $(\inputAlphabetElement^\blocklength, \outputAlphabetElement^\blocklength)$ under a channel $\channel$ with input distribution $\inputDistribution$}
}
\newglossaryentry{codebook}{
  name={\ensuremath{\codebook=\codebook(\codewordIndex)_{\codewordIndex=1}^{\exp(\blocklength\codebookRate)}}},
  text={\ensuremath{\codebook}},
  description={codebook}
}
\newglossaryentry{eveMarginalOutputDistribution}{
  name={\ensuremath{\eveMarginalOutputDistribution{\channel^\blocklength,\codebook}}},
  description={distribution of the output of channel $\channel^\blocklength$ if a uniformly random code word from $\codebook$ is transmitted}
}
\newglossaryentry{absolutelyContinuous}{
  name={\ensuremath{\generalMeasureOne \absolutelyContinuous \generalMeasureTwo}},
  description={measure $\generalMeasureOne$ is absolutely continuous with respect to measure $\generalMeasureTwo$},
}
\newglossaryentry{RNDerivative}{
  name={\ensuremath{\RNDerivative{\generalMeasureOne}{\generalMeasureTwo}}},
  description={Radon-Nikodym derivative of $\generalMeasureOne$ with respect to $\generalMeasureTwo$}
}
\newglossaryentry{eveOutputRV}{
  name={\ensuremath{\eveOutputRV}},
  description={eavesdropper's channel output}
}
\newglossaryentry{totalvariation}{
  name={\ensuremath{\totalvariation{\cdot}}},
  description={total variation norm on the vector space of signed, finite measures}
}
\newglossaryentry{exp}{
  name={\ensuremath{\exp(\cdot)}},
  description={exponentiation with Euler's number as basis}
}
\newglossaryentry{log}{
  name={\ensuremath{\log(\cdot)}},
  description={natural logarithm}
}
\newglossaryentry{noiseRV}{
  name={\ensuremath{\noiseRV}},
  description={additive channel noise}
}
\newglossaryentry{stdnormalcdf}{
  name={\ensuremath{\stdnormalcdf}},
  description={cumulative distribution function of the standard normal distribution}
}
\newglossaryentry{stdnormalpdf}{
  name={\ensuremath{\stdnormalpdf}},
  description={probability density function of the standard normal distribution}
}

\title{Towards Secure Over-The-Air Computation}
\author{\IEEEauthorblockN{ Matthias Frey,\IEEEauthorrefmark{1} Igor Bjelakovi\'c,\IEEEauthorrefmark{2}\IEEEauthorrefmark{3} and S\l awomir~Sta\'{n}czak\IEEEauthorrefmark{2}\IEEEauthorrefmark{3}
    } %
    \\
  \IEEEauthorblockA{
    \IEEEauthorrefmark{1}Department of Electrical and Electronic Engineering, The University of Melbourne, Australia,
    \IEEEauthorrefmark{2}Technische Universität Berlin, Germany,
    and
    \IEEEauthorrefmark{3}Fraunhofer Heinrich Hertz Institute, Berlin, Germany
  }%
}

\begin{document}

\maketitle
\blankfootnote{
Part of this work was presented at the 2021 IEEE International Symposium on Information Theory, 12-20 July 2021, Melbourne, Victoria, Australia.

This work was supported by the German Research Foundation (DFG)
within their priority program SPP 1798 ``Compressed Sensing in Information Processing'' and under grants STA 864/7-1 and STA 864/15-1. This work was also supported by the Federal Ministry of Education and Research of Germany in the program of ``Souverän. Digital. Vernetzt.''. Joint project 6G-RIC, project identification numbers: 16KISK020K, 16KISK030.
}

\begin{abstract}
We propose a new method to protect \gls{ota} computation schemes against passive eavesdropping. Our method uses a friendly jammer whose signal is -- contrary to common intuition -- stronger at the legitimate receiver than it is at the eavesdropper. We focus on the computation of arithmetic averages over an \gls{awgn} channel. The derived secrecy guarantee translates to a lower bound on the eavesdropper's mean square error while the question of how to provide operationally more significant guarantees such as semantic security remains open for future work. The key ingredients in proving the security guarantees are a known result on channel resolvability and a generalization of existing achievability results on coding for compound channels.
\end{abstract}


\section{Introduction}
\label{sec:introduction}
In many envisioned applications in wireless networks, the receiver requires only a function of values available at the distributed transmitters rather than the full information about the values themselves. Examples include distributed Federated Learning~\cite{amiri2020machine} and distributed anomaly detection in sensor networks~\cite{kiril}. In such cases, analog \gls{ota} computation schemes can deliver sizable performance gains over classical separation-based approaches especially when the number of transmitters is large~\cite{gastpar2003source,goldenbaum2013robust,kiril,bjelakovic2019distributed,amiri2020machine,liu2020over,frey2021over}.

In some of the foreseen scenarios, such as e-health or industrial applications, security and privacy are expected to be major concerns in addition to efficient resource usage. Information theoretic secrecy can complement classic cryptography in addressing these issues. A natural way to enhance security in  an analog \gls{ota} computation setting such as~\cite{gastpar2003source,goldenbaum2013robust,kiril,bjelakovic2019distributed,amiri2020machine,liu2020over,frey2021over} is to add a jammer to the system that deteriorates the signal-to-noise ratio (SNR) of the eavesdropper and thereby prevents it from reconstructing a low-noise estimate of the objective function. In this case, it is necessary to place the jammer so that its signal is significantly stronger at the eavesdropper than it is at the legitimate receiver. Since in general, the exact position of an eavesdropper is not known, jammers typically have to be placed at multiple locations. In this work, we propose to turn this situation around and place the jammer so that its signal is stronger at the legitimate receiver than it is at the eavesdropper. Such a setup is often easier to realize since the jammer can for instance be set up in proximity to the legitimate receiver and in certain settings, such as factory buildings, it may be feasible to assume that the attacker is located, e.g., outside of the building while the legitimate receiver and the transmitters are inside. Our proposed scheme operates under the assumption that the jamming signal is stronger at the legitimate receiver than it is at the eavesdropper. It is applicable to the special case of \gls{ota} computation of an arithmetic average over an \gls{awgn} channel. The main idea is to carefully construct a jamming signal in such a way that it can be fully reconstructed (and therefore canceled in post-processing) by the legitimate receiver, while the eavesdropper is impacted by the jamming signal as though it was white noise.

\subsection{Prior Work}
\label{sec:literature}
To the best of our knowledge, the \gls{ota} computation problem over a wiretap channel has not yet been considered in the literature. Therefore, in this subsection we briefly summarize the literature on the building blocks we use for the approach to the wiretap \gls{ota} computation channel that we propose in this work as well as for literature on concepts that are closely related to the ones presented in this paper. 

\paragraph{\texorpdfstring{\gls{ota}}{OTA} computation}
The concept of analog \gls{ota} computation was originally introduced in~\cite{gastpar2003source} and further developed in~\cite{kiril,goldenbaum2013robust}. In~\cite{bjelakovic2019distributed,frey2021over}, we revisited this idea, adapted the existing scheme and provided an extension to a large class of functions and analysis of the estimation error for finite block length in a very general, fast-fading setting. There is also a digital version of \gls{ota} computation in which domain and range of the computed versions are finite, which was introduced in~\cite{nazer2007computation}. There are many more prior works in this area. For details, we refer the reader to the literature section in~\cite{frey2021over}.

\paragraph{Coding for compound channels}
The compound channel problem was introduced independently in \cite{dobrushin1959optimum,blackwell1959capacity,wolfowitz1959simultaneous}, while first independent results for the capacity expression can be found in \cite{blackwell1959capacity,wolfowitz1959simultaneous}. These works, however, explore mainly the case of finite input and output alphabets. The \emph{semi-continuous} case in which only the input alphabet is assumed to be finite is briefly touched upon in \cite{wolfowitz1959simultaneous} and studied in more detail in~\cite{kesten1961some} which provides an example showing that the capacity expression from the finite case does not carry over to the semi-continuous case in general. The semi-continuous case was further explored in~\cite{yoshihara1965coding,ahlswede1967certain}. In many cases of practical interest, the capacity expression from the finite case can be generalized to the \emph{continuous} case in which neither input nor output alphabets are assumed to be finite, as was found in~\cite{root1968capacity} for a class of Gaussian compound channels. Wiretap compound channels with finite alphabets are studied in~\cite{bjelakovic2013secrecy}. Gaussian compound wiretap channels and related models have been investigated in~\cite{he2014mimo}. However, the compound channel part in this work focuses on continuous-alphabet extensions of point-to-point compound channels.

\paragraph{Channel Resolvability and Semantic Security}
The concept of channel resolvability was introduced in~\cite{wyner1975common,han1993approximation}. Further results relevant in the context of this work appeared, e.g., in~\cite{csiszar1996almost,devetak2005private,hayashi2016secure,cuff2016soft}. We use our generalization~\cite{frey2018resolvability} for continuous channels as a basis for our proposed scheme. Although we cannot provide full semantic security guarantees in this work, we also heavily draw from the idea of obtaining semantic security by means of channel resolvability, which is developed in~\cite{hayashi2006general,csiszar2011information,BellareSemantic,BlochStrongSecrecy}.

\paragraph{Friendly Jamming}
The idea of friendly jamming has been used in~\cite{negi2005secret} to aid a transmitter-receiver pair in protecting a point-to-point transmission from a passive eavesdropper. Distributed and centralized beamforming techniques are used so that the jamming signal impacts the signal-to-noise ratio at the eavesdropper but not at the legitimate receiver. Several more recent works (cf., e.g., \cite{vilela2010friendly,vilela2011wireless,stanojev2012improving}) have expanded upon this idea and refined the friendly jamming techniques. In the context of two-way wiretap channels,~\cite{tekin2008general,pierrot2011strongly,he2013role} use cooperative jamming, in which the transmitter/receiver nodes add artificial noise to their wiretap-encoded messages. In~\cite{pierrot2011strongly}, channel resolvability is used to prove strong secrecy guarantees for such schemes. To the best of our knowledge, there are no prior works which use jamming to protect \gls{ota} computation against eavesdropping.

\paragraph{Physical Layer Security}
The concept of information theoretic secrecy was introduced in~\cite{ShannonSecrecy} and the wiretap channel model together with a weaker, but more tractable notion of secrecy was introduced in~\cite{WynerWiretap}. Based on this, various stronger secrecy notions have been introduced and investigated (e.g., \cite{MaurerStrongSecret,HouEffectiveSecrecy,BellareSemantic}). All of these existing works investigate how digitally coded transmissions can be protected against eavesdropping, while in the present work, we focus on uncoded analog transmissions over multiple-access channels.

\paragraph{Computational Wiretap Channels}
\cite{doliveira2018computational,bassi2019mutual} study a system model in which a function computation is to be protected from an eavesdropper. Contrary to this work, there is only one transmitter, and the eavesdropper has the same channel output as the legitimate receiver. The security guarantee hinges upon the eavesdropper wanting to compute a function that is different from the receiver's intended function, and one key application that is noted by the authors is therefore information-theoretic privacy.

\subsection{Summary of the Main Contributions and Outline}
\label{sec:contributions}
The main contributions of this paper are as follows:
\begin{enumerate}
 \item \label{item:secrecy} We propose a novel framework and result for incorporating security considerations into the \gls{ota} computation of an arithmetic average over an \gls{awgn} channel. In this framework, a friendly jammer is included in the system which deteriorates the eavesdropper's SNR while not significantly impacting the legitimate receiver's ability to obtain an approximation of the function value which is to be \gls{ota} computed.
 \item In order to approach this problem, we observe a connection between the secure \gls{ota} computation problem and the problems of compound channel coding and channel resolvability for point-to-point channels. This connection is not dependent on the \gls{awgn} channel model and may therefore be useful also to establish results for more general channel models.
 \item We prove a theorem on compound channel coding for continuous alphabets. It is a generalization of the result of the part of~\cite{root1968capacity} which considers finite-dimensional Gaussian channels, and we can consequently recover this result as a special case.
\end{enumerate}

In Section~\ref{sec:secure-ota-c-awgn}, we state and prove our main result about \gls{ota} computation of an arithmetic mean over an \gls{awgn} channel. Part of the proof relies on technical results from later sections and is therefore deferred to Section~\ref{sec:general-channel}. In Section~\ref{sec:compound}, we state and prove the point-to-point compound channel coding theorem that is required in the following section. In Section~\ref{sec:general-channel}, we give the full details of the connection between the secure \gls{ota} computation problem, compound channel coding and channel resolvability that is used to establish the result of Section~\ref{sec:secure-ota-c-awgn}. Section~\ref{sec:conclusion} concludes the paper and states open questions for future research.

Throughout the paper, we define notation where it is first used. For the reader's convenience, a summary of notational symbols can be found in Fig.~\ref{fig:symbols}.

\begin{figure*}
\renewcommand{\glossarysection}[2][]{}
\glsfindwidesttoplevelname
\renewcommand{\glstreenamefmt}{}
\printnoidxglossary[sort=use, type=main, style=alttree]
\caption{Table of symbols.}
\label{fig:symbols}
\end{figure*}

\section{System Model and Main Result}
\label{sec:secure-ota-c-awgn}
In this section, we introduce the detailed system and channel model, and then proceed to state and discuss the main result of this paper. The part of the proof that requires the technical tools of later sections is only sketched here, while the full technical details are deferred to Section~\ref{sec:general-channel}.

\subsection{Distributed Function Approximation with Jamming (DFA-J)}
\label{sec:dfaj}
In the following, we introduce the system model for \acrshort{dfaj} which is an extension of the model used in~\cite{bjelakovic2019distributed}.

Let $\analogMessageAlphabet{1}, \dots, \analogMessageAlphabet{\aliceNum}$ be measurable spaces. The goal is to approximate functions \gls{objectiveFunction} over a multiple-access channel \gls{channel} with measurable input alphabets \gls{channelInAlph} and a measurable output alphabet \gls{channelOutAlph} in a distributed setting. An admissible \emph{\acrfull{dfa}} scheme for $\objectiveFunction: \analogMessageAlphabet{1} \times \ldots \times \analogMessageAlphabet{\aliceNum}\to\reals$ for $\blocklength$ channel uses is a pair $(\alicePreproc{}^\blocklength, \bobPostproc^\blocklength)$, consisting of:

\begin{enumerate}
\item A pre-processing function \gls{alicePreproc} for the transmitters \gls{transmitters}, where each $\alicePreproc{\aliceIndex}^\blocklength$ is of the form
\[
\alicePreproc{\aliceIndex}^\blocklength(\analogMessageAlphabetElement{\aliceIndex})
=
(
  \alicePreprocElement_{\aliceIndex,\blockIndex}(
    \analogMessageAlphabetElement{\aliceIndex},
    \alicePreprocRandomness_\aliceIndex(\blockIndex)
  )
)_{\blockIndex=1}^{\blocklength}\in \inputAlphabet_\aliceIndex^{\blocklength}
\]
with i.i.d. random variables $\alicePreprocRandomness_\aliceIndex(1), \ldots, \alicePreprocRandomness_\aliceIndex(\blocklength)$ and a measurable map 
\[
(
  \analogMessageAlphabetElement{\aliceIndex},
  \alicePreprocInnerVariable_1,
  \ldots,
  \alicePreprocInnerVariable_\blocklength
)
\mapsto
(
  \alicePreprocElement_{\aliceIndex,\blockIndex}(
    \analogMessageAlphabetElement{\aliceIndex},
    \alicePreprocInnerVariable_\blockIndex 
  )
)_{\blockIndex=1}^{\blocklength}\in \inputAlphabet_\aliceIndex^{\blocklength}.
\]
\item A post-processing function \gls{bobPostproc} for the receiver \gls{receiver}: The receiver is allowed to apply a measurable recovery function $\bobPostproc^\blocklength: \inputAlphabet^\blocklength \to \reals$ upon observing the output of the channel.
\end{enumerate}
Note that this, contrary to the system model in~\cite{bjelakovic2019distributed}, imposes the restriction that the pre-processing is i.i.d. across channel uses, which will be crucial for the security extension to the approximation scheme. Although this definition of admissible schemes is slightly less general, the scheme proposed in~\cite{bjelakovic2019distributed} still is an admissible scheme even in this stricter sense.

So in order to approximate $\objectiveFunction$, the transmitters apply their pre-processing maps to
\[
(\analogMessageAlphabetElement{1},\ldots, \analogMessageAlphabetElement{\aliceNum})
\in
\analogMessageAlphabet{1}\times \ldots \times \analogMessageAlphabet{\aliceNum}
\]
resulting in $\alicePreproc{1}^\blocklength (\analogMessageAlphabetElement{1}), \ldots, \alicePreproc{\aliceNum}^\blocklength(\analogMessageAlphabetElement{\aliceNum})$, which are sent over the channel. 
The receiver observes the output of the channel and applies the recovery map $\bobPostproc^\blocklength$. The whole process defines an estimate \gls{objectiveEstimator} of $\objectiveFunction$.

Depending on the application at hand, there are multiple ways in which the quality of the estimate $\objectiveEstimator$ can be quantified.

\begin{definition}
\label{def:approximation-quality}
\begin{enumerate}
 \item \label{item:tail-approximation} Let $\functionApproximationError,\tail \in (0,1)$ and $\objectiveFunction$ be given. We say that $\objectiveFunction$ is $\functionApproximationError$-approximated after $\blocklength$ channel uses with confidence level $\tail$ if there is an approximation scheme
$(\alicePreproc{}^\blocklength,\bobPostproc^\blocklength)$ such that the resulting estimate $\objectiveEstimator$ of $\objectiveFunction$ satisfies
\[
\Probability\left(
  \absolute{
    \objectiveEstimator
    -
    \objectiveFunction(\analogMessageAlphabetElement{1}, \dots, \analogMessageAlphabetElement{\aliceNum})
  }
  \ge
  \functionApproximationError
\right)
\le
\tail
\]
for all $s^K:= (s_1, \ldots , s_K) \in \analogMessageAlphabet{1}\times \ldots \times \analogMessageAlphabet{K} $.
 \item We say that $\objectiveFunction$ is $\mseError$-\acrshort{mse}-approximated if we have
\[
\Expectation\left(
  \left(
    \objectiveEstimator
    -
    \objectiveFunction(\analogMessageAlphabetElement{1}, \dots, \analogMessageAlphabetElement{\aliceNum})
  \right)^2
\right)
\leq
\mseError,
\]
where the expectation is over the joint distribution of $\stateSpaceElement_1, \dots, \stateSpaceElement_\aliceNum$ and $\objectiveEstimator$ which is induced by the \gls{dfa} scheme and the channel.
\end{enumerate}
\end{definition}

\begin{figure}
\centering
\begin{tikzpicture}
\node[rectangle,draw,minimum width=.8cm]            (inK)          at (0,0) {$\alice{\aliceNum}$};
\node[rectangle,draw,minimum width=.8cm]            (in2)          at (0,2) {$\alice{2}$};
\node[rectangle,draw,minimum width=.8cm]            (in1)          at (0,3) {$\alice{1}$};
\node[rectangle,draw]                               (encK)         at (2,0)  {$\alicePreproc{\aliceNum}^\blocklength$};
\node[rectangle,draw]                               (enc2)         at (2,2)  {$\alicePreproc{2}^\blocklength$};
\node[rectangle,draw]                               (enc1)         at (2,3)  {$\alicePreproc{1}^\blocklength$};
\node[rectangle,draw,minimum width=.8cm]            (jammer)       at (0,-1) {$\jammer$};
\node[rectangle,minimum height=4.5cm,align=center,draw] (channel) at (4,1) {$\blocklength$-fold\\channel};
\node[rectangle,draw]                               (dec)          at (6,2.5) {$\bobPostproc^\blocklength$};
\node[rectangle,draw]                               (out)          at (8,2.5) {$\bob$};
\node[rectangle,draw]                               (eve_out)      at (8,-.5) {$\eve$};
\node                                               (vdots)        at (.8,1.2) {\Shortstack{. . . . . .}};

\draw[->] (inK) -- (encK) node[midway,above] {$\analogMessageAlphabetElement{\aliceNum}$};
\draw[->] (in2) -- (enc2) node[midway,above] {$\analogMessageAlphabetElement{2}$};
\draw[->] (in1) -- (enc1) node[midway,above] {$\analogMessageAlphabetElement{1}$};

\draw[->] (encK) -- (encK-|channel.west) node[midway,above] {$\processedAnalogMessageRV{\aliceNum}^\blocklength$};
\draw[->] (enc2) -- (enc2-|channel.west) node[midway,above] {$\processedAnalogMessageRV{2}^\blocklength$};
\draw[->] (enc1) -- (enc1-|channel.west) node[midway,above] {$\processedAnalogMessageRV{1}^\blocklength$};
\draw[->] (jammer) -- (jammer-|channel.west) node[midway,above] {$\inputRV^\blocklength$};

\draw[->] (dec-|channel.east) -- (dec) node[midway,above] {$\outputRV^\blocklength$};
\draw[->] (dec) -- (out) node[midway,above] {$\objectiveEstimator$};

\draw[->] (eve_out-|channel.east) -- (eve_out) node[midway,above] {$\eveOutputRV^\blocklength$};

\draw[red] (1.2,3.7) -- (4.9,3.7) node[midway,above]{$\effectiveChannelBob^\blocklength$} -- (4.9,0) -- (4,0) -- (4,-1.4) -- (1.2,-1.4) -- (1.2,3.7);
\draw[cyan] (1.3,3.6) -- (4,3.6) -- (4,2) -- (4.8,2) -- (4.8,-1.5) -- (1.3,-1.5) node[midway, below]{$\effectiveChannelEve^\blocklength$} -- (1.3,3.6);
\end{tikzpicture}
\caption{System model for \gls{dfaj} described in Section~\ref{sec:dfaj}.}
\label{fig:systemmodel}
\end{figure}

In this work, we extend the \gls{dfa} system model adding an attacker \gls{eavesdropper} which attempts to eavesdrop on the transmission and wants to gain knowledge about $\analogMessageAlphabetElement{1}, \dots, \analogMessageAlphabetElement{\aliceNum}$. At each channel use, $\eve$ observes an output \gls{eveOutputRV} ranging over the eavesdropper's alphabet $\gls{eveAlph}$. As a counter-measure, we add a friendly jammer \gls{jammer} which transmits some jamming sequence $\gls{inputRV}^\blocklength$ with the objective to prevent $\eve$ from obtaining information while still allowing $\bob$ to obtain a good estimate of $\objectiveFunction(\analogMessageAlphabetElement{1}, \dots, \analogMessageAlphabetElement{\aliceNum})$. This extended model is depicted in Fig.~\ref{fig:systemmodel}.

\begin{definition}
\label{def:dfa-jamming}
A scheme for \acrfull{dfaj} consists of:
\begin{itemize}
 \item a \gls{dfa} scheme; i.e., pre- and post-processing schemes, and
 \item a jamming strategy given by a probability distribution on $\gls{inAlph}^\blocklength$.
\end{itemize}
We say that a \gls{dfaj} scheme \emph{allows reconstruction of the jamming signal with probability $\errorProb$} if there is a decoding function $\jammingDecoder: \outputAlphabet^\blocklength \rightarrow \inputAlphabet^\blocklength$ such that
\[
\sup_{\stateSpaceElement_1 \in \stateSpace_1, \dots, \stateSpaceElement_\aliceNum \in \stateSpace_\aliceNum}
\Probability_{\stateSpaceElement_1, \dots, \stateSpaceElement_\aliceNum}\left(
  \jammingDecoder(\gls{outputRV}^\blocklength) \neq \inputRV^\blocklength
\right)
\leq
\errorProb
\]
and $\errorProb$ is the smallest number with this property.
\end{definition}

The objective is to find admissible pre- and post-processing strategies as well as a jamming strategy such that $\bob$ can obtain a good approximation $\objectiveEstimator$ of $\objectiveFunction(\analogMessageAlphabetElement{1}, \dots, \analogMessageAlphabetElement{\aliceNum})$ while bounding the usefulness of any information that $\eve$ can obtain about $\analogMessageAlphabetElement{1}, \dots, \analogMessageAlphabetElement{\aliceNum}$.

Together with the channel, a \gls{dfaj} scheme induces a probability distribution $\eveOutputDistributionGeneral{\stateSpaceElement_1, \dots, \stateSpaceElement_\aliceNum}$ on $\eveAlphabet^\blocklength$ for each $(\stateSpaceElement_1, \dots, \stateSpaceElement_\aliceNum) \in \stateSpace_1 \times \ldots \times \stateSpace_\aliceNum$. How secure the scheme is depends on how strongly $\eveOutputDistributionGeneral{\stateSpaceElement_1, \dots, \stateSpaceElement_\aliceNum}$ depends on $\stateSpaceElement_1, \dots, \stateSpaceElement_\aliceNum$. In the following, we formalize this notion.

Any measurable function $\eveObjective: \stateSpace_1 \times \ldots \times \stateSpace_\aliceNum \rightarrow \eveObjectiveRange$, where $\eveObjectiveRange$ is a measurable space, is called an \emph{eavesdropper's objective}.
\begin{definition}
\label{def:security}
\begin{enumerate}
\item Given a real number $\securityError \geq 0$, we say that a \gls{dfaj} scheme is \emph{$\securityError$-semantically secure} if there is a probability measure $\uselessDistribution$ on $\eveAlphabet^\blocklength$ such that for all $(\stateSpaceElement_1, \dots, \stateSpaceElement_\aliceNum) \in \stateSpace_1 \times \ldots \times \stateSpace_\aliceNum$,
\begin{align}
\label{eq:semantic-security}
\totalvariationlr{
  \eveOutputDistributionGeneral{\stateSpaceElement_1, \dots, \stateSpaceElement_\aliceNum}
  -
  \uselessDistribution
}
\leq
\securityError,
\end{align}
where \gls{totalvariation} denotes the total variation norm on finite signed measures. The probability measure $\uselessDistribution$ can be arbitrary here except for the requirement that is is independent of $\stateSpaceElement_1, \dots, \stateSpaceElement_\aliceNum$.

\item Let $\eveObjective: \stateSpace_1 \times \ldots \times \stateSpace_\aliceNum \rightarrow \eveObjectiveRange$, where $\eveObjectiveRange \subseteq \reals$ is measurable and bounded, be an eavesdropper's objective. Let $\secLoss \geq 0$ be a real number. We say that a \gls{dfaj} scheme is \emph{$(\eveObjective,\secLoss)$-\acrshort{mse}-secure} if under a uniform distribution of $\eveObjective(\stateSpaceElement_1, \dots, \stateSpaceElement_\aliceNum)$, for every estimator $\eveDecoder: \eveAlphabet^\blocklength \rightarrow \eveObjectiveRange$, we have
\begin{align*}
\Expectation\left(
  \big(
    \eveDecoder(\eveOutputRV^\blocklength)
    -
    \eveObjective(\stateSpaceElement_1, \dots, \stateSpaceElement_\aliceNum)
  \big)^2
\right)
\geq
\secLoss,
\end{align*}
where the expectation is over the joint distribution of $\stateSpaceElement_1, \dots, \stateSpaceElement_\aliceNum$ and $\outputRV^\blocklength$ which results from the application of the \gls{dfaj} scheme and the channel.
\end{enumerate}
\end{definition}
If a scheme is $(\eveObjective,\secLoss)$-\acrshort{mse}-secure, it means that any estimator used by the eavesdropper has a \gls{mse} of no less than $\secLoss$ in the case of a uniformly distributed objective. This means that $\stateSpaceElement_1, \dots, \stateSpaceElement_\aliceNum$ are randomly distributed in such a way that $\eveObjective(\stateSpaceElement_1, \dots, \stateSpaceElement_\aliceNum)$ follows a uniform distribution which implies that $\stateSpaceElement_1, \dots, \stateSpaceElement_\aliceNum$ cannot be i.i.d. uniform in general. Our motivation for assuming a uniform objective instead of uniform i.i.d. $\stateSpaceElement_1, \dots, \stateSpaceElement_\aliceNum$ is the following: For many choices of $\eveObjective$ that we consider relevant (and in particular the computation of arithmetic mean on which we will focus in this paper), the function $\eveObjective(\stateSpaceElement_1, \dots, \stateSpaceElement_\aliceNum)$ tends to concentrate at its expectation for large values of $\aliceNum$ if $\stateSpaceElement_1, \dots, \stateSpaceElement_\aliceNum$ are independent. Since the statistical information is assumed to be known at the eavesdropper, it could therefore achieve low \gls{mse} even without intercepting any channel output. However, only the results for the \gls{awgn} channel in this section rely on such an assumption of uniformity while the technical results in Section~\ref{sec:general-channel} assume that $\stateSpaceElement_1, \dots, \stateSpaceElement_\aliceNum$ are deterministic, but arbitrary. This means that the results of Section~\ref{sec:general-channel} specialize to arbitrary stochastic models of $\stateSpaceElement_1, \dots, \stateSpaceElement_\aliceNum$, and in particular to the uniform and non-independent case. Therefore, there is hope that this somewhat restrictive assumption could be lifted in future research.

In a sense made explicit by the following lemma, semantic security is the stronger of the two security notions from Definition~\ref{def:security}.
\begin{lemma}
\label{lemma:semanticstronger}
Let $\eveObjectiveRange := [\intervalLowerBound, \intervalUpperBound]$, let $\eveObjective: \stateSpace_1 \times \ldots \times \stateSpace_\aliceNum \rightarrow \eveObjectiveRange$ be an eavesdropper's objective and $\securityError \geq 0$ a real number. Assume that $\eveObjective(\stateSpaceElement_1, \dots, \stateSpaceElement_\aliceNum)$ is uniformly distributed on $\eveObjectiveRange$.

Then, any \gls{dfaj} scheme that is $\securityError$-semantically secure is also $(\eveObjective,(1/12-\securityError)(\intervalUpperBound-\intervalLowerBound)^2)$-\acrshort{mse}-secure.
\end{lemma}
\begin{proof}
Let $\eveDecoder: \eveAlphabet^\blocklength \rightarrow \eveObjectiveRange$. Then, assuming the distribution of $\stateSpaceElement_1, \dots, \stateSpaceElement_\aliceNum$ corresponds to a uniform distribution on $[\intervalLowerBound,\intervalUpperBound]$ of $\eveObjective(\stateSpaceElement_1, \dots, \stateSpaceElement_\aliceNum)$, we have
\begin{align*}
\Expectation_{\stateSpaceElement_1, \dots, \stateSpaceElement_\aliceNum}
\Expectation_{\eveOutputDistributionGeneral{\stateSpaceElement_1, \dots, \stateSpaceElement_\aliceNum}}
  \left(
    \left(
      \eveDecoder(\eveOutputRV^\blocklength)
      -
      \eveObjective(\stateSpaceElement_1, \dots, \stateSpaceElement_\aliceNum)
    \right)^2
  \right)
&=
\Expectation_{\stateSpaceElement_1, \dots, \stateSpaceElement_\aliceNum}
\int_0^{(\intervalUpperBound-\intervalLowerBound)^2}
  \eveOutputDistributionGeneral{\stateSpaceElement_1, \dots, \stateSpaceElement_\aliceNum}
    \bigg(
      \Big(
        \eveDecoder(\eveOutputRV^\blocklength)
        -
        \eveObjective(\stateSpaceElement_1, \dots, \stateSpaceElement_\aliceNum)
      \Big)^2
      >
      \generalReal
    \bigg)
d \generalReal
\\
\overset{(\ref{eq:semantic-security})}&{\geq}
\Expectation_{\stateSpaceElement_1, \dots, \stateSpaceElement_\aliceNum}
\int_0^{(\intervalUpperBound-\intervalLowerBound)^2}\bigg(
  \uselessDistribution
    \Big(
      \left(
        \eveDecoder(\eveOutputRV^\blocklength)
        -
        \eveObjective(\stateSpaceElement_1, \dots, \stateSpaceElement_\aliceNum)
      \right)^2
      >
      \generalReal
    \Big)
  -
  \securityError
\bigg)
d \generalReal
\\
&=
\Expectation_{\stateSpaceElement_1, \dots, \stateSpaceElement_\aliceNum}
\Expectation_\uselessDistribution
  \left(
    \left(
      \eveDecoder(\eveOutputRV^\blocklength)
      -
      \eveObjective(\stateSpaceElement_1, \dots, \stateSpaceElement_\aliceNum)
    \right)^2
  \right)
-\securityError(\intervalUpperBound-\intervalLowerBound)^2
\\
\overset{(a)}&{\geq}
\Expectation_{\stateSpaceElement_1, \dots, \stateSpaceElement_\aliceNum}
\Expectation_\uselessDistribution
  \left(
    \left(
      \Expectation_\uselessDistribution
        \eveObjective(\stateSpaceElement_1, \dots, \stateSpaceElement_\aliceNum)
      -
      \eveObjective(\stateSpaceElement_1, \dots, \stateSpaceElement_\aliceNum)
    \right)^2
  \right)
-\securityError(\intervalUpperBound-\intervalLowerBound)^2
\\
\overset{(b)}&{=}
\left(
  \frac{1}{12} - \securityError
\right)
(\intervalUpperBound-\intervalLowerBound)^2,
\end{align*}
where the step (a) is because under $\uselessDistribution$, $\eveOutputRV^\blocklength$ is independent of $\stateSpaceElement_1, \dots, \stateSpaceElement_\aliceNum$, and therefore the \gls{mse} is minimized by the mean of $\eveObjective(\stateSpaceElement_1, \dots, \stateSpaceElement_\aliceNum)$. Finally, step (b) follows from the assumption that $\eveObjective(\stateSpaceElement_1, \dots, \stateSpaceElement_\aliceNum)$ is uniform on $[\intervalLowerBound,\intervalUpperBound]$, and hence its variance is known (see, e.g.,~\cite[Example 3.4]{bertsekas2008introduction}).
\end{proof}

\subsection{\texorpdfstring{\gls{awgn}}{AWGN} Channel Model}
\label{sec:AWGN}
In general, the approximation scheme even without an eavesdropper or jammer highly depends on the particular structure of the channel and $\objectiveFunction$. In this work, we focus on the computation of arithmetic means over \gls{awgn} channels (although some of the technical results we develop for this purpose hold under more general conditions). Specifically, the objective function is given as
\begin{align}
\label{eq:awgnObjective}
\objectiveFunction:
(\analogMessageAlphabetElement{1}, \dots, \analogMessageAlphabetElement{\aliceNum})
\mapsto
\frac{1}{\aliceNum}\sum_{\aliceIndex=1}^\aliceNum \analogMessageAlphabetElement{\aliceIndex},
\end{align}
where for all $\aliceIndex$, $\analogMessageAlphabet{\aliceIndex} = [-1,1]$. The channel is given by
\begin{align}
\label{eq:awgnbob}
\outputRV
&=
\fadingAliceBob \sum_{\aliceIndex=1}^\aliceNum \gls{processedAnalogMessageRV}
+
\fadingJammerBob \inputRV
+
\noiseBob
\\
\label{eq:awgneve}
\eveOutputRV
&=
\fadingAliceEve \sum_{\aliceIndex=1}^\aliceNum \processedAnalogMessageRV{\aliceIndex}
+
\fadingJammerEve \inputRV
+
\noiseEve.
\end{align}
$\noiseBob$ is centered normal with variance $\stddevBob^2$ and $\noiseEve$ is centered normal with variance $\stddevEve^2$. The real channel coefficients $\fadingAliceBob, \fadingJammerBob, \fadingAliceEve, \fadingJammerEve$ are assumed deterministic and known everywhere. The channel is used $\blocklength$ times with transmitter input sequences $\processedAnalogMessageRV{\aliceIndex}^\blocklength$ for each $\aliceIndex \in \{1, \dots, \aliceNum\}$ and $\inputRV^\blocklength$ for the jammer. The input sequences are subject to the average power constraints
\[
\frac{1}{\blocklength}
\sum_{\blockIndex=1}^\blocklength
  \left(\processedAnalogMessageRV{\aliceIndex,\blockIndex}\right)^2
\leq
\alicePowerConstraint,
~~
\frac{1}{\blocklength}
\sum_{\blockIndex=1}^\blocklength
  (\inputRV_\blockIndex)^2
\leq
\jammerPowerConstraint.
\]

\subsection{Main Result}
\begin{figure}
\centering
\begin{tikzpicture}
\begin{axis}[
  xlabel={$\stddev^2$},
  ylabel={$\mseFunctionModified(\stddev)$},
  ymax=.4,
]
\addplot[mark=none] table [x=var, y=mse, col sep=comma] {mse_function_plot.csv};
\draw[dashed] ({rel axis cs:1,0}|-{axis cs:0,0.333}) -- ({rel axis cs:0,0}|-{axis cs:0,0.333});
\end{axis}
\end{tikzpicture}
\caption{Illustration of the \gls{mse} guarantees of Theorem~\ref{theorem:awgn-scheme-real}. The dashed line is the \gls{mse} which an eavesdropper would have without any received signal (i.e., guessing the middle of the interval).}
\label{fig:mseFunction}
\end{figure}
\begin{theorem}
\label{theorem:awgn-scheme-real}
Consider the wiretap channel given by (\ref{eq:awgnbob}) and (\ref{eq:awgneve}) and the objective function $\objectiveFunction$ defined in (\ref{eq:awgnObjective}). Assume that $\stateSpaceElement_1, \dots, \stateSpaceElement_\aliceNum$ are distributed in such a way that $\objectiveFunction(\stateSpaceElement_1, \dots, \stateSpaceElement_\aliceNum)$ is uniform in $[-1,1]$. Define
\begin{equation}
\label{eq:effective-noise}
\stddevBobEff^2
:=
\frac{\stddevBob^2}{\fadingAliceBob^2 \aliceNum^2 \alicePowerConstraint}
,~~
\stddevEveEff^2
:=
\frac{\stddevEve^2 + \fadingJammerEve^2\jammerPowerConstraint}
     {\fadingAliceEve^2 \aliceNum^2 \alicePowerConstraint}
\end{equation}
and
\begin{equation}
\label{eq:mseFunction}
\mseFunction(\generalReal)
:=
\int_0^\generalReal
  \int_{-\infty}^\infty
    \left(
      \generalRealThree
      +
      \frac{\stdnormalpdf(-\generalRealThree) - \stdnormalpdf(\generalReal - \generalRealThree)}
           {\stdnormalcdf(\generalReal - \generalRealThree) - \stdnormalcdf(-\generalRealThree)}
      -
      \generalRealTwo
    \right)^2
    \cdot
    \frac{1}{\generalReal}
    \stdnormalpdf(\generalRealTwo - \generalRealThree)
  d \generalRealThree
d \generalRealTwo,
\end{equation}
where \gls{stdnormalpdf} denotes the probability density function and \gls{stdnormalcdf} the cumulative distribution function of the standard normal distribution, respectively.
Assume that the channel from $\jammer$ to $\bob$ is stronger than the channel from $\jammer$ to $\eve$, i.e., $\fadingJammerBob/\stddevBob > \fadingJammerEve/\stddevEve$. Then there is a \gls{dfaj} scheme and there are constants $\finalconstOne, \finalconstTwo > 0$ such that for sufficiently large $\blocklength$, the following hold:
\begin{itemize}
 \item $\bob$ can approximate the objective function $\objectiveFunction(\analogMessageAlphabetElement{1}, \dots, \analogMessageAlphabetElement{\aliceNum})$ with a \gls{mse} not exceeding
 \begin{align}
 \label{eq:awgn-scheme-real-bob}
 \stddevBobEff^2 \mseFunction\left(\frac{2}{\stddevBobEff}\right)
 +
 \exp(-\blocklength\finalconstOne)
 \end{align}
 \item The scheme is $(\objectiveFunction,\secLoss)$-\acrshort{mse}-secure, where
 \begin{align}
 \label{eq:awgn-scheme-real-eve}
 \secLoss
 :=
 \stddevEveEff^2 \mseFunction\left(\frac{2}{\stddevEveEff}\right)
 -
 \exp(-\blocklength\finalconstTwo).
 \end{align}
\end{itemize}
\end{theorem}

Here and for the remainder of the paper, the functions \gls{exp} and \gls{log} both use Euler's number as their basis. In order to understand the impact of the function $\mseFunction$ that appears in the security and approximation guarantees, we refer the reader to the plot of the function $\mseFunctionModified: \stddev \mapsto \stddev^2 \mseFunction(2/\stddev)$ in Fig.~\ref{fig:mseFunction}.

The proof of Theorem~\ref{theorem:awgn-scheme-real} is divided into two main steps:
\begin{itemize}
 \item In Lemma~\ref{lemma:awgn-scheme-ideal}, we examine a scenario where the jammer employs white noise as a jamming signal. The instantaneous channel inputs of the jammer are known to $\bob$ but not to $\eve$. We establish \gls{mse}-security and \gls{mse}-approximation guarantees for this case.
 \item We show that there is a jamming strategy induced by a suitable code book which guarantees that $\bob$ can reconstruct the jammer's channel input with high probability, but for $\eve$, the jamming strategy resembles the case of white noise. The \gls{mse}-security and \gls{mse}-approximation guarantees then follow via a comparison to the case of Lemma~\ref{lemma:awgn-scheme-ideal}. This part of the proof is deferred to the end of Section~\ref{sec:general-channel}, after all necessary technical ingredients are introduced. The reconstruction of the jamming signal at $\bob$ is based on the compound channel coding result in Section~\ref{sec:compound}, and the resemblance of white noise at $\eve$ is based on a known channel resolvability result.
\end{itemize}

\begin{lemma}
\label{lemma:awgn-scheme-ideal}
Consider the wiretap channel given by (\ref{eq:awgnbob}) and (\ref{eq:awgneve}) and the objective function $\objectiveFunction$ defined in (\ref{eq:awgnObjective}). Assume that $\objectiveFunction(\stateSpaceElement_1, \dots, \stateSpaceElement_\aliceNum)$ is uniformly distributed on $\eveObjectiveRange$. Furthermore, suppose that the jamming sequence $\inputRV^\blocklength$ is i.i.d. centered Gaussian with variance $\jammerPowerConstraint$ and that it is known at the legitimate receiver while the eavesdropper has only statistical information. Then, under the definitions (\ref{eq:effective-noise}) and (\ref{eq:mseFunction}), there is a \gls{dfaj} scheme which is $(\objectiveFunction, \stddevEveEff^2 \mseFunction(2/\stddevEveEff))$-\acrshort{mse}-secure and $(\stddevBobEff^2 \mseFunction(2/\stddevBobEff))$-\acrshort{mse}-approximates $\objectiveFunction$ at the receiver.
\end{lemma}

The proof of Lemma~\ref{lemma:awgn-scheme-ideal} is based on a few facts from statistics. We only state the relevant lemmas here. Since they are straightforward consequences of elementary known facts about minimum \gls{mse} estimators, we expect that they are folklore in the field of statistics, however, we are not aware of any reference that states these facts in the form in which we need them for our proof. Therefore, we include the proofs of the following two lemmas in the appendix for the sake of completeness.

\begin{lemma}
\label{lemma:bayes-estimator}
If $\generalrvOne$ is distributed uniformly on $[\intervalLowerBound, \intervalUpperBound]$ and, conditioned on $\generalrvOne$, $\generalrvTwo_1, \dots, \generalrvTwo_\blocklength$ are i.i.d. normally distributed with mean $\generalrvOne$ and variance $\stddev^2$, then the minimum \gls{mse} estimator for estimating $\generalrvOne$ from the observations $\generalrvTwo_1, \dots, \generalrvTwo_\blocklength$ is
\begin{align}
\label{eq:bayes-estimator}
\hat{\generalrvOne}
:=
\bar{\generalrvTwo}
+
\frac{\stddev}{\sqrt{\blocklength}}
\cdot
\frac{\stdnormalpdf\left(\frac{\intervalLowerBound-\bar{\generalrvTwo}}{\stddev/\sqrt{\blocklength}}\right)
      -
      \stdnormalpdf\left(\frac{\intervalUpperBound-\bar{\generalrvTwo}}{\stddev/\sqrt{\blocklength}}\right)}
     {\stdnormalcdf\left(\frac{\intervalUpperBound-\bar{\generalrvTwo}}{\stddev/\sqrt{\blocklength}}\right)
      -
      \stdnormalcdf\left(\frac{\intervalLowerBound-\bar{\generalrvTwo}}{\stddev/\sqrt{\blocklength}}\right)},
\end{align}
where $\bar{\generalrvTwo} := \frac{1}{\blocklength}\sum_{\blockIndex=1}^{\blocklength} \generalrvTwo_\blockIndex$.
\end{lemma}

\begin{lemma}
\label{lemma:bayes-estimator-mse}
Under the assumptions of Lemma~\ref{lemma:bayes-estimator}, the estimator $\hat{\generalrvOne}$ satisfies
\[
\Expectation\left(
  \left(
    \generalrvOne
    -
    \hat{\generalrvOne}
  \right)^2
\right)
=
\frac{\stddev^2}{\blocklength}
\mseFunction\left(
  \frac{\intervalUpperBound - \intervalLowerBound}
       {\stddev/\sqrt{\blocklength}}
\right),
\]
with $\mseFunction$ as defined in (\ref{eq:mseFunction}).
\end{lemma}

\begin{proof}[Proof of Lemma~\ref{lemma:awgn-scheme-ideal}]
We use the following transmission strategy:
\begin{align}
\label{eq:awgnJammingStrategy}
\inputRV_\blockIndex:~ &\text{Gaussian with mean } 0 \text{ and variance } \jammerPowerConstraint,
\\
\label{eq:awgnPreproc}
\alicePreproc{\aliceIndex}^\blocklength:~ &\analogMessageAlphabetElement{\aliceIndex} \mapsto (1, \dots, 1) \cdot \analogMessageAlphabetElement{\aliceIndex} \sqrt{\frac{\alicePowerConstraint}{\blocklength}}
\end{align}
The receiver can obtain
\begin{align*}
\outputRV_\blockIndex'
&:=
\frac{\outputRV_\blockIndex - \fadingJammerBob \inputRV_\blockIndex}
     {\fadingAliceBob \aliceNum \sqrt{\alicePowerConstraint/\blocklength}}
\\
&=
\frac{
        \fadingAliceBob \sum_{\aliceIndex=1}^\aliceNum \processedAnalogMessageRV{\aliceIndex}
        +
        \noiseBobComponent{\blockIndex}
     }
     {\fadingAliceBob \aliceNum \sqrt{\alicePowerConstraint/\blocklength}}
\\
&=
\frac{
        \fadingAliceBob \sum_{\aliceIndex=1}^\aliceNum \analogMessageAlphabetElement{\aliceIndex} \sqrt{\alicePowerConstraint/\blocklength}
        +
        \noiseBobComponent{\blockIndex}
     }
     {\fadingAliceBob \aliceNum \sqrt{\alicePowerConstraint/\blocklength}}
\\
&=
\objectiveFunction(\analogMessageAlphabetElement{1}, \dots, \analogMessageAlphabetElement{\aliceNum})
+
\frac{
        \noiseBobComponent{\blockIndex}
     }
     {\fadingAliceBob \aliceNum \sqrt{\alicePowerConstraint/\blocklength}}.
\end{align*}
We define the post-processing operation $\bobPostproc^\blocklength$ at the receiver as first obtaining $\outputRV_1', \dots, \outputRV_\blocklength'$ and then computing the minimum \gls{mse} estimator from Lemma~\ref{lemma:bayes-estimator}. With this choice, Lemma~\ref{lemma:bayes-estimator-mse} yields the claimed reconstruction error guarantee.

On the other hand, the output at $\eve$ is given by
\begin{align*}
\eveOutputRV_\blockIndex
&\stackrel{(\ref{eq:awgneve})}{=}
\fadingAliceEve \sum_{\aliceIndex=1}^\aliceNum \processedAnalogMessageRV{\aliceIndex}
+
\fadingJammerEve \inputRV_\blockIndex
+
\noiseEveComponent{\blockIndex}
\\
&\stackrel{(\ref{eq:awgnPreproc})}{=}
\fadingAliceEve \sum_{\aliceIndex=1}^\aliceNum \analogMessageAlphabetElement{\aliceIndex} \sqrt{\alicePowerConstraint/\blocklength}
+
\fadingJammerEve \inputRV_\blockIndex
+
\noiseEveComponent{\blockIndex}
\\
&\stackrel{(\ref{eq:awgnObjective})}{=}
\objectiveFunction(\analogMessageAlphabetElement{1}, \dots, \analogMessageAlphabetElement{\aliceNum})
\cdot
\aliceNum \sqrt{\alicePowerConstraint/\blocklength} \fadingAliceEve
+
\fadingJammerEve \inputRV_\blockIndex
+
\noiseEveComponent{\blockIndex}.
\end{align*}
From this, Lemmas~\ref{lemma:bayes-estimator} and \ref{lemma:bayes-estimator-mse} yield the claimed \acrshort{mse}-security of the scheme.
\end{proof}

\subsection{Special case $\aliceNum=1$}
\label{sec:onetransmitter}
We conclude this section with a brief discussion of the important special case $\aliceNum=1$. While one of the main motivations of the methods developed in this paper is their scalability to large values of $\aliceNum$, the case of low values of $\aliceNum$ can also be interesting in many practical applications and be instructive to understand the nature of our results better.

For the special case of only a single transmitter ($\aliceNum=1$), the problem reduces to a point-to-point transmission of the real number $\objectiveFunction(\stateSpaceElement_1)$ in the presence of an eavesdropper and a friendly jammer. In our results in this paper, there is no assumption that $\aliceNum$ has to be large; in particular, they remain applicable also when $\aliceNum=1$. However, since in this case no function of \emph{distributed} values has to be computed over the channel, it is possible to separately source and channel encode $\objectiveFunction(\stateSpaceElement_1)$. After the source coding step has been performed, the remaining problem is very similar to jammer-aided secret communication as treated for instance in~\cite{negi2005secret,vilela2010friendly,vilela2011wireless}.

But although this approach is applicable to the same communication task, it is important to note that the way in which the friendly jammer has to be placed differs significantly. In the approach of this paper, the jamming signal has to be stronger at the legitimate receiver than it is at the eavesdropper. As long as this condition is satisfied, the legitimate receiver has the ability to almost completely cancel the jamming signal. This means that our method remains applicable even if the gap in terms of jammer signal strength between the legitimate receiver and the eavesdropper is relatively small. In~\cite{negi2005secret,vilela2010friendly,vilela2011wireless}, on the other hand, it is necessary that the jamming signal is stronger at the eavesdropper than it is at the legitimate receiver. Moreover, this gap between signal strengths has to be as large as possible since the jammer's signal strength at the legitimate receiver diminishes the capacity of the main channel. Therefore, our results in this case are more suitable for scenarios where is possible to assure a high jamming signal strength at the legitimate receiver while results from~\cite{negi2005secret,vilela2010friendly,vilela2011wireless} are more suitable in cases where all possible eavesdropper locations can be covered with strong jamming signals that have very low strength at the location of the legitimate receiver.

With respect to the open research questions given in Section~\ref{sec:conclusion}, we remark that methods from the literature can be used to achieve semantic security with slight adaptations; such a construction is for instance sketched in~\cite{utkovski2019learning}. We are not aware of practically feasible schemes that achieve semantic security, but we expect that weaker guarantees such as \gls{mse} security could be derived, e.g., for the approach given in~\cite{klinc2011ldpc}. In order to accommodate a friendly jammer in the system model, all of these approaches would need to be combined with the works on friendly jamming discussed above. Therefore, it would remain necessary to also have the assumption that the jamming signal is significantly stronger at the eavesdropper than it is at the legitimate receiver. For the case in which this assumption is reversed as in the present work, to the best of our knowledge these questions remain open even for $\aliceNum=1$.

\section{Coding for the Compound Channel}
\label{sec:compound}
In this section, we state and prove a coding result for compound channels with continuous alphabets. This result is used in the proof of Theorem~\ref{theorem:jammersystem} which in turn is a technical contribution needed to prove Theorem~\ref{theorem:awgn-scheme-real}. Although similar to results already available in the literature, it is slightly more general and may therefore also be of independent interest.

\subsection{System Model and Preliminary Definitions}
We begin with some preliminary notations and definitions. Given measures $\generalMeasureOne$ and $\generalMeasureTwo$, we say that $\generalMeasureOne$ is absolutely continuous with respect to $\generalMeasureTwo$, or \gls{absolutelyContinuous}, if all $\generalMeasureTwo$-null sets are also $\generalMeasureOne$-null sets. If we have $\generalMeasureOne \absolutelyContinuous \generalMeasureTwo$, then the Radon-Nikodym derivative \gls{RNDerivative} exists, which is an (up to a $\generalMeasureTwo$-null set) uniquely determined function with the property $\int_\generalSet \RNDerivative{\generalMeasureOne}{\generalMeasureTwo}d\generalMeasureTwo = \generalMeasureOne(\generalSet)$ for all measurable sets $\generalSet$.

For any channel $\channel$, we denote the joint input-output distribution under $\inputDistribution$ and $\channel$ by \gls{jointDist} and the marginal for $\outputAlphabet$ by \gls{marginalDist}. Since we use Euler's number as the basis of the functions $\exp$ and $\log$, all of the information quantities defined in the following are given in nats. We define the \emph{information density} of tuples of elements of the input and output alphabets under the channel $\channel$ and an input distribution $\inputDistribution$ as
\[
\gls{informationDensity}
:=
\log
\RNDerivative{\channel^\blocklength(\inputAlphabetElement^\blocklength, \cdot)}
             {\marginalOutputDistribution{\inputDistribution}{\channel}^\blocklength}
(\outputAlphabetElement^\blocklength).
\]
By convention, if $\channel^\blocklength(\inputAlphabetElement^\blocklength, \cdot) \notAbsolutelyContinuous \marginalOutputDistribution{\inputDistribution}{\channel}^\blocklength$, the information density is $\infty$. 

Correspondingly, the \emph{mutual information} is defined as
\[
\gls{mutInf}
:=
\Expectation_{\inputOutputDistribution{\inputDistribution}{\channel}}
  \informationDensity{\inputDistribution}{\channel}{\inputRV}{\outputRV}.
\]
Although the integrand is guaranteed to not be $\pm \infty$ on non-null sets, the mutual information integral can be infinite. Moreover, given two probability measures $\generalMeasureOne$ and $\generalMeasureTwo$, we define the \emph{Rényi divergence} of order $\renyiParam \in (0,1) \cup (1,\infty)$ between them as
\[
\gls{renyidiv}
:=
\frac{1}{\renyiParam - 1}
\log
\Expectation_\generalMeasureOne\left(
  \left(
    \RNDerivative{\generalMeasureOne}{\generalMeasureTwo}
  \right)^{\renyiParam-1}
\right).
\]
Again, by convention, the Rényi divergence is $\infty$ if $\generalMeasureOne \notAbsolutelyContinuous \generalMeasureTwo$.
$
\gls{kldiv}
:=
\lim_{\renyiParam \nearrow 1}
\renyidiv{\renyiParam}{\generalMeasureOne}{\generalMeasureTwo}
$
is the \emph{Kullback-Leibler divergence}.

A \emph{compound channel} is a family \gls{compoundChannel} of memoryless time-discrete point-to-point channels with common input alphabet \gls{inAlph} and output alphabet $\outputAlphabet$. The transmitter's channel input is passed through a fixed $\compoundChannel{\stateSpaceElement}$ for the entire block length, but the transmitter does not control the choice of $\stateSpaceElement$, nor is it governed by a probability distribution. In this work, we assume neither the transmitter nor the receiver knows $\stateSpaceElement$. A \emph{compound channel code} with block length $\blocklength$ and rate $\codebookRate$ consists of an encoder $\genericEncoder: \{1, \dots, \exp(\blocklength\codebookRate)\} \rightarrow \inputAlphabet^\blocklength$ and a decoder $\genericDecoder: \inputAlphabet^\blocklength \rightarrow \{1, \dots, \exp(\blocklength\codebookRate)\}$. We say that it has error probability $\errorProb$ if under a uniform distribution of $\gls{messageRV} \in \{1, \dots, \exp(\blocklength\codebookRate)\}$, the following is true: Let $\outputRV^\blocklength$ be constructed by passing the components of $\inputRV^\blocklength := \genericEncoder(\messageRV)$ independently through $\compoundChannel{\stateSpaceElement}$. Then, we have
\[
\sup_{\stateSpaceElement \in \stateSpace}
\Expectation_\messageRV
\Probability_\stateSpaceElement(\codewordIndex \neq \genericDecoder(\outputRV^\blocklength)) \leq \errorProb,
\]
where $\errorProb$ is the smallest number with this property.

Our proof of Theorem~\ref{theorem:awgn-scheme-real} hinges on coding for a particular class of Gaussian compound channels. Such channels have continuous input and output alphabets, so we need an achievability result for compound channels with continuous input and output alphabets. As mentioned in Section~\ref{sec:introduction}, it is shown in~\cite{kesten1961some} that even in the case that only the output alphabet is countably infinite, the capacity expressions from the finite case~\cite{blackwell1959capacity,wolfowitz1959simultaneous} do not carry over. It is therefore clear that an additional assumption on the compound channel is needed. In existing literature (e.g., \cite{blackwell1959capacity,root1968capacity}), the problem is often approached by proving that the compound channel can be approximated by a finite class of channels in which case classical channel coding techniques such as joint typicality decoding can be adapted in a straightforward manner. In this work, we choose to directly pose the approximability of the compound channel by a finite class of channels as an assumption of our coding theorem.

\begin{definition}
\label{def:approximation}
Given a compound channel $(\compoundChannel{\stateSpaceElement})_{\stateSpaceElement \in \stateSpace}$ with input alphabet $\inputAlphabet$ and output alphabet $\outputAlphabet$, we say that it can be \emph{$(\channelApproximationError,\generalChannelApproximationNumber)$-approximated} under a probability distribution $\inputDistribution$ on $\inputAlphabet$ if there is a sequence \gls{approximateChannel} of channels from $\inputAlphabet$ to $\outputAlphabet$ such that for every $\stateSpaceElement \in \stateSpace$, there is $\channelApproximationIndex \in \{1,\dots,\generalChannelApproximationNumber\}$ such that
\begin{align}
\label{eq:approximation-kldiv1}
&\Expectation_\inputDistribution
  \kldiv{\compoundChannel{\stateSpaceElement}(\inputRV,\cdot)}
        {\approximateChannel{\channelApproximationError}{\channelApproximationIndex}(\inputRV,\cdot)}
  \leq
  \channelApproximationError
\\
\label{eq:approximation-renyidiv}
\exists \renyiParam > 1~\forall \inputAlphabetElement \in \inputAlphabet:~
  &\renyidiv{\renyiParam}
            {\compoundChannel{\stateSpaceElement}(\inputAlphabetElement,\cdot)}
            {\approximateChannel{\channelApproximationError}{\channelApproximationIndex}(\inputAlphabetElement,\cdot)}
  <
  \infty
\\
\label{eq:approximation-absolute-continuity}
\forall \inputAlphabetElement \in \inputAlphabet:~
  &\approximateChannel{\channelApproximationError}{\channelApproximationIndex}(\inputAlphabetElement,\cdot)
  \absolutelyContinuous
  \compoundChannel{\stateSpaceElement}(\inputAlphabetElement,\cdot)
\\
\label{eq:approximation-information1}
&\information{\inputDistribution}{\approximateChannel{\channelApproximationError}{\channelApproximationIndex}}
-
\information{\inputDistribution}{\compoundChannel{\stateSpaceElement}}
\leq
\channelApproximationError,
\end{align}
and for every $\channelApproximationIndex \in \{1,\dots,\generalChannelApproximationNumber\}$ there is $\stateSpaceElement \in \stateSpace$ such that
\begin{align}
\label{eq:approximation-information2}
\information{\inputDistribution}{\compoundChannel{\stateSpaceElement}}
-
\information{\inputDistribution}{\approximateChannel{\channelApproximationError}{\channelApproximationIndex}}
\leq
\channelApproximationError.
\end{align}
\end{definition}
Conditions (\ref{eq:approximation-kldiv1}) and (\ref{eq:approximation-information1}) tell us in what sense the approximating channel must be similar to the approximated channel, while the remaining conditions are of a more technical nature. (\ref{eq:approximation-renyidiv}) and (\ref{eq:approximation-absolute-continuity}) ensure that all moment-generating functions and Radon-Nikodym derivatives that we use to derive the exponential error bounds exist. Finally, (\ref{eq:approximation-information2}) tells us that in a certain sense, the approximating sequence cannot be too rich, and it can usually be ensured that it holds by not including unnecessary channels in the sequence.

\subsection{Feasibility of Channel Approximation}
\label{sec:chapprox-feasibility}
In this subsection, we provide some tools and examples to argue that many compound channels of practical interest can indeed be $(\channelApproximationError,\generalChannelApproximationNumber)$-approximated so that Theorem~\ref{theorem:compound} may be applied to them. In particular, the results in this section imply that Theorem~\ref{theorem:compound} can be applied to the class of Gaussian channels we need to prove Theorem~\ref{theorem:awgn-scheme-real}.

We begin with an observation that shows that the approximability criterion of Definition~\ref{def:approximation} is a generalization of the assumption of finite channel alphabets that is used in~\cite{blackwell1959capacity,wolfowitz1959simultaneous}.
\begin{remark}
\label{remark:compound-finite}
\cite[Lemma 4]{blackwell1959capacity} implies that for every compound channel $(\compoundChannel{\stateSpaceElement})_{\stateSpaceElement \in \stateSpace}$ with finite input and output alphabets and every $\channelApproximationError > 0$, there is an integer $\channelApproximationNumber{\channelApproximationError}$ such that $(\compoundChannel{\stateSpaceElement})_{\stateSpaceElement \in \stateSpace}$ can be $(\channelApproximationError,\channelApproximationNumber{\channelApproximationError})$-approximated.

We repeat the construction here and discuss how this fact is proved.

Let $\finiteApproxInteger$ be an integer which satisfies
\[
\finiteApproxInteger
\geq
\max\left(
  \frac{4\cardinality{\outputAlphabet}^3}{\channelApproximationError^2},
  \frac{2\cardinality{\outputAlphabet}^2}{\channelApproximationError}
\right).
\]
Given $\stateSpaceElement \in \stateSpace$, we construct a channel $\correspondingApproximateChannel{\stateSpaceElement}$. To this end, given any $\inputAlphabetElement \in \inputAlphabet$, we fix an enumeration $(\outputAlphabetElement_\generalSummationIndex)_{\generalSummationIndex=1}^{\cardinality{\outputAlphabet}}$ such that the finite sequence $(\compoundChannel{\stateSpaceElement}(\inputAlphabetElement,\{\outputAlphabetElement_\generalSummationBound\}))_{\generalSummationIndex=1}^{\cardinality{\outputAlphabet}}$ is nondecreasing. For every $\generalSummationIndex < \cardinality{\outputAlphabet}$, we can then uniquely choose a value for $\correspondingApproximateChannel{\stateSpaceElement}(\inputAlphabetElement,\{\outputAlphabetElement_\generalSummationBound\})$ such that it is an integer multiple of $1/\finiteApproxInteger$ and
\begin{align}
\label{eq:finite-approximation}
\compoundChannel{\stateSpaceElement}(\inputAlphabetElement,\{\outputAlphabetElement_\generalSummationBound\})
\leq
\correspondingApproximateChannel{\stateSpaceElement}(\inputAlphabetElement,\{\outputAlphabetElement_\generalSummationBound\})
<
\compoundChannel{\stateSpaceElement}(\inputAlphabetElement,\{\outputAlphabetElement_\generalSummationBound\}) + \frac{1}{\finiteApproxInteger}.
\end{align}
It is argued in \cite{blackwell1959capacity} that this leaves a positive probability mass for $\correspondingApproximateChannel{\stateSpaceElement}(\inputAlphabetElement,\{\outputAlphabetElement_{\cardinality{\outputAlphabet}}\})$ and therefore, this construction fully defines a channel $\correspondingApproximateChannel{\stateSpaceElement}$. We define the approximation sequence $(\approximateChannel{\channelApproximationError}{\channelApproximationIndex})_{\channelApproximationIndex=1}^{\channelApproximationNumber{\channelApproximationError}}$ as an enumeration of the set $\{\correspondingApproximateChannel{\stateSpaceElement}: \stateSpaceElement \in \stateSpace\}$. The cardinality of this set is upper bounded by $(\finiteApproxInteger + 1)^{\cardinality{\inputAlphabet} \cardinality{\outputAlphabet}}$ since all singleton probabilities are integer multiples of $1/\finiteApproxInteger$.

For finite alphabets, (\ref{eq:approximation-renyidiv}) is trivially satisfied since Rényi divergence is in this case always finite~\cite{vanerven2014renyi}. Regarding the absolute continuity criterion (\ref{eq:approximation-absolute-continuity}), we recall that $\correspondingApproximateChannel{\stateSpaceElement}(\inputAlphabetElement,\{\outputAlphabetElement_{\cardinality{\outputAlphabet}}\})$ always has a positive probability, and for $\generalSummationIndex < \cardinality{\outputAlphabet}$, the assumption $\compoundChannel{\stateSpaceElement}(\inputAlphabetElement,\{\outputAlphabetElement_\generalSummationBound\}) = 0$ immediately implies $\correspondingApproximateChannel{\stateSpaceElement}(\inputAlphabetElement,\{\outputAlphabetElement_{\cardinality{\outputAlphabet}}\})=0$ by (\ref{eq:finite-approximation}), since $0$ is the only integer multiple of $1/\finiteApproxInteger$ which is strictly smaller than $1/\finiteApproxInteger$. The proof in~\cite{blackwell1959capacity} exploits (\ref{eq:finite-approximation}) to prove that the absolute difference between the information of $\compoundChannel{\stateSpaceElement}$ and $\correspondingApproximateChannel{\stateSpaceElement}$ under any input distribution is at most $2\cardinality{\outputAlphabet}^{3/2}\finiteApproxInteger^{-1/2}$ (statement (c) of the lemma) which by our choice of $\finiteApproxInteger$ immediately implies (\ref{eq:approximation-information1}) and (\ref{eq:approximation-information2}). Moreover, it is shown that (\ref{eq:finite-approximation}) also implies that for all $\inputAlphabetElement \in \inputAlphabet, \outputAlphabetElement \in \outputAlphabet$,
\[
\log\frac{\compoundChannel{\stateSpaceElement}(\inputAlphabetElement,\{\outputAlphabetElement\})}{\correspondingApproximateChannel{\stateSpaceElement}(\inputAlphabetElement,\{\outputAlphabetElement\})}
\leq
\frac{2\cardinality{\outputAlphabet}^2}{\finiteApproxInteger}
\]
(statement (b) of the lemma) which by our choice of $\finiteApproxInteger$ implies (\ref{eq:approximation-kldiv1}).

\end{remark}
For many channels of interest, $(\channelApproximationError,\channelApproximationNumber{\channelApproximationError})$-approximability can be shown directly by going through properties (\ref{eq:approximation-kldiv1}) -- (\ref{eq:approximation-information2}). However, it is  often easier to make an argument involving topological properties of $\stateSpace$. The following lemma provides some machinery to this end. In its statement, we use the following generalization of continuity of functions: Let $\stateSpace$ be a topological space. A function $\generalFunction: \stateSpace \rightarrow [-\infty,\infty]$ is called \emph{upper semi-continuous} at $\stateSpaceElement_0$ if for every $\generalReal_0 > \generalFunction(\stateSpaceElement_0)$ there exists an open set $\generalSet \subseteq \stateSpace$ with $\stateSpaceElement_0 \in \generalSet$ and for all $\stateSpaceElement \in \generalSet$, $\generalReal_0 > \generalFunction(\stateSpaceElement)$. $\generalFunction$ is called \emph{lower semi-continuous} at $\stateSpaceElement_0$ if $-\generalFunction$ is upper semi-continuous at $\stateSpaceElement_0$. $\generalFunction$ is called \emph{upper semi-continuous (lower semi-continuous)} if it is upper semi-continuous (lower semi-continuous) at every point of its domain.
\begin{lemma}
\label{lemma:topological-approximation}
Let $(\compoundChannel{\stateSpaceElement})_{\stateSpaceElement \in \stateSpace}$ be a compound channel with input alphabet $\inputAlphabet$ and output alphabet $\outputAlphabet$, let $\inputDistribution$ be a probability distribution on $\inputAlphabet$ and assume that there is a topology on $\stateSpace$ such that $\stateSpace$ is compact and
\begin{align}
\label{eq:topological-approximation-kldiv}
&\forall \stateSpaceElement_0 \in \stateSpace:~
\stateSpaceElement
\mapsto
\Expectation_\inputDistribution \kldiv{\compoundChannel{\stateSpaceElement}(\inputRV,\cdot)}{\compoundChannel{\stateSpaceElement_0}(\inputRV,\cdot)}
\text{is upper semi-continuous at $\stateSpaceElement_0$,}
\\
\label{eq:topological-approximation-renyidiv}
&\forall \stateSpaceElement_1, \stateSpaceElement_2 \in \stateSpace~ \exists \renyiParam > 1~ \forall \inputAlphabetElement \in \inputAlphabet: 
\renyidiv{\renyiParam}{\compoundChannel{\stateSpaceElement_1}(\inputAlphabetElement,\cdot)}{\compoundChannel{\stateSpaceElement_2}(\inputAlphabetElement,\cdot)}
<
\infty,
\\
\label{eq:topological-approximation-information}
&\stateSpaceElement \mapsto \information{\inputDistribution}{\compoundChannel{\stateSpaceElement}}
\text{ is lower semi-continuous.}
\end{align}
Then, for any $\channelApproximationError > 0$, there is $\channelApproximationNumber{\channelApproximationError}$ such that $(\compoundChannel{\stateSpaceElement})_{\stateSpaceElement \in \stateSpace}$ can be $(\channelApproximationError,\channelApproximationNumber{\channelApproximationError})$-approximated under $\inputDistribution$.
\end{lemma}
\begin{proof}
Fix some $\channelApproximationError > 0$. For a given $\stateSpaceElement \in \stateSpace$, consider
\[
\{
  \stateSpaceElement':
  \Expectation_\inputDistribution \kldiv{\compoundChannel{\stateSpaceElement'}(\inputRV,\cdot)}{\compoundChannel{\stateSpaceElement}(\inputRV,\cdot)}
  <
  \channelApproximationError
\}
\cap
\{
  \stateSpaceElement':
  \information{\inputDistribution}{\compoundChannel{\stateSpaceElement}}
  -
  \information{\inputDistribution}{\compoundChannel{\stateSpaceElement'}}
  <
  \channelApproximationError
\}.
\]

Clearly, (\ref{eq:topological-approximation-kldiv}) and (\ref{eq:topological-approximation-information}) ensure that this intersection is a neighborhood of $\stateSpaceElement$, so we can find an open neighborhood $\coverElement_\stateSpaceElement$ contained in it. Thus, $(\coverElement_\stateSpaceElement)_{\stateSpaceElement \in \stateSpace}$ is an open cover of $\stateSpace$ and therefore, the compactness of $\stateSpace$ yields a finite subcover $\coverElement_{\stateSpaceElement_1}, \dots, \coverElement_{\stateSpaceElement_\channelApproximationNumber{\channelApproximationError}}$. We set
$
\approximateChannel{\channelApproximationError}{\channelApproximationIndex} := \compoundChannel{\stateSpaceElement_\channelApproximationIndex}
$
and given any $\stateSpaceElement \in \stateSpace$, we choose $\channelApproximationIndex$ such that $\stateSpaceElement \in \coverElement_{\stateSpaceElement_\channelApproximationIndex}$ and argue that $\approximateChannel{\channelApproximationError}{\channelApproximationIndex}$ satisfies  (\ref{eq:approximation-kldiv1}), (\ref{eq:approximation-renyidiv}) and (\ref{eq:approximation-information1}). To this end, we note that (\ref{eq:approximation-renyidiv}) and (\ref{eq:approximation-absolute-continuity}) follow from (\ref{eq:topological-approximation-renyidiv}), while (\ref{eq:approximation-kldiv1}) and (\ref{eq:approximation-information1}) are ensured by the definition of $\coverElement_{\stateSpaceElement_\channelApproximationIndex}$. Finally, (\ref{eq:approximation-information2}) is trivially satisfied, concluding the proof.
\end{proof}

We now make use of Lemma~\ref{lemma:topological-approximation} to prove that a large class of Gaussian fading multiple-input and multiple-output channels can actually be $(\channelApproximationError,\channelApproximationNumber{\channelApproximationError})$-approximated and thus Theorem~\ref{theorem:compound} can be applied to them. The class of compound channels covered in the following theorem contains the class considered in~\cite[Sections 3 and 4]{root1968capacity} as a proper subset. We denote the set of symmetric, positive semidefinite $\generalSummationBound \times \generalSummationBound$-matrices with \gls{psdsym} and the set of symmetric, positive definite $\generalSummationBound \times \generalSummationBound$-matrices with \gls{pdsym}.

\begin{theorem}
\label{theorem:gaussian-compound-approximation}
Let $\inputAlphabet = \reals^\numTxAntennas$, $\outputAlphabet = \reals^\numRxAntennas$, let $\stateSpace$ be a compact subset of $\reals^{\numRxAntennas\numTxAntennas} \times \psdSymMatrices{\numRxAntennas\numTxAntennas} \times \reals^\numRxAntennas \times \pdSymMatrices{\numRxAntennas}$ (under the topology induced by the Frobenius norm). For any $\stateSpaceElement = (\mean_\channelMatrix, \covarianceMatrix_\channelMatrix, \mean_\noiseRV, \covarianceMatrix_\noiseRV) \in \stateSpace$, let $\compoundChannel{\stateSpaceElement}$ be the channel given by
\[
\outputRV = \channelMatrix \inputRV + \noiseRV,
\]
where the channel input $\inputRV$ has range $\reals^\numTxAntennas$, the channel output $\outputRV$ has range $\reals^\numRxAntennas$, the entries of the $\numRxAntennas \times \numTxAntennas$ fading matrix $\channelMatrix$ follow a multivariate normal distribution with mean $\mean_\channelMatrix$ and covariance matrix $\covarianceMatrix_\channelMatrix$ and the additive noise \gls{noiseRV} is independent of $\channelMatrix$ and follows a multivariate normal distribution with mean $\mean_\noiseRV$ and covariance matrix $\covarianceMatrix_\noiseRV$. Let $\inputDistribution$ be a distribution on $\inputAlphabet$ and assume that either $\inputDistribution$ is a multivariate Gaussian with positive definite covariance matrix or that the support of $\inputDistribution$ is contained in some compact set.
Then, given any $\channelApproximationError > 0$, there is $\channelApproximationNumber{\channelApproximationError}$ such that $(\compoundChannel{\stateSpaceElement})_{\stateSpaceElement \in \stateSpace}$ can be $(\channelApproximationError,\channelApproximationNumber{\channelApproximationError})$-approximated under $\inputDistribution$.
\end{theorem}

\begin{proof}
We show that the conditions of Lemma~\ref{lemma:topological-approximation} are met. \cite{gil2011renyi} provides closed-form expressions for Rényi and Kullback-Leibler divergences between multivariate normal distributions. The only fact that we are going to use and which is apparent from these expressions, however, is that the Rényi and Kullback-Leibler divergences between two multivariate normal distributions are finite and continuous in the mean vectors and covariance matrices of the distributions wherever the covariance matrices are positive definite or, equivalently, both distributions are absolutely continuous with respect to the Lebesgue measure.

$\covarianceMatrix_\noiseRV \in \pdSymMatrices{\numRxAntennas}$ and therefore, given any $\inputAlphabetElement \in \inputAlphabet$, $\compoundChannel{\stateSpaceElement}(\inputAlphabetElement, \cdot)$ is absolutely continuous with respect to the Lebesgue measure and thus has a positive definite covariance matrix and a density $\density{\compoundChannel{\stateSpaceElement}(\inputAlphabetElement, \cdot)}$, which implies (\ref{eq:topological-approximation-renyidiv}).

Next, from the well-known closed-form expression of the multivariate normal density, we know that for any $\inputAlphabetElement$ and $\outputAlphabetElement$, $\density{\compoundChannel{\stateSpaceElement}(\inputAlphabetElement, \cdot)}(\outputAlphabetElement)$ is continuous in $\stateSpaceElement$. The boundedness of $\stateSpace$ implies a uniform upper bound on $\density{\compoundChannel{\stateSpaceElement}(\inputAlphabetElement, \cdot)}(\outputAlphabetElement)$, so we can use the theorem of dominated convergence to argue that the marginal density $\density{\marginalOutputDistribution{\inputDistribution}{\compoundChannel{\stateSpaceElement}}}(\outputAlphabetElement) = \Expectation_\inputDistribution \density{\compoundChannel{\stateSpaceElement}(\inputRV, \cdot)}(\outputAlphabetElement)$ depends continuously on $\stateSpaceElement$ for any fixed $\outputAlphabetElement$. We write
\[
\information{\inputDistribution}{\compoundChannel{\stateSpaceElement}}
=
\Expectation_{\inputDistribution \marginalOutputDistribution{\inputDistribution}{\compoundChannel{\stateSpaceElement}}}\left(
  \frac{\density{\compoundChannel{\stateSpaceElement}(\inputRV, \cdot)}(\outputRV)}
       {\density{\marginalOutputDistribution{\inputDistribution}{\compoundChannel{\stateSpaceElement}}}(\outputRV)}
  \log
    \frac{\density{\compoundChannel{\stateSpaceElement}(\inputRV, \cdot)}(\outputRV)}
         {\density{\marginalOutputDistribution{\inputDistribution}{\compoundChannel{\stateSpaceElement}}}(\outputRV)}
\right).
\]
Since the integrand is lower bounded by $-\exp(-1)$, (\ref{eq:topological-approximation-information}) follows as an application of Fatou's lemma.

Finally, in order to argue (\ref{eq:topological-approximation-kldiv}), we distinguish between the two cases in the statement of the theorem. 

First, suppose that there is a compact subset $\hat{\inputAlphabet} \subseteq \inputAlphabet$ with $\inputDistribution(\inputAlphabet \setminus \hat{\inputAlphabet}) = 0$. For any fixed $\stateSpaceElement_0$, the map
\[ (\stateSpaceElement,\inputAlphabetElement) \mapsto \kldiv{\compoundChannel{\stateSpaceElement}(\inputAlphabetElement, \cdot)}{\compoundChannel{\stateSpaceElement_0}(\inputAlphabetElement, \cdot)} \]
is continuous, therefore the image of $\stateSpace \times \hat{\inputAlphabet}$ is compact and hence bounded. We can therefore invoke the theorem of dominated convergence and argue that (\ref{eq:topological-approximation-kldiv}) is satisfied.

Now, suppose that $\inputDistribution$ is multivariate Gaussian with positive definite covariance matrix. We write
\begin{align*}
\Expectation_\inputDistribution \kldiv{\compoundChannel{\stateSpaceElement}(\inputRV,\cdot)}{\compoundChannel{\stateSpaceElement_0}(\inputRV,\cdot)}
&=
\Expectation_{\inputDistribution}
\Expectation_{\compoundChannel{\stateSpaceElement}(\inputRV, \cdot)}
\log\frac{\density{\inputDistribution}(\inputRV)\density{\compoundChannel{\stateSpaceElement}(\inputRV,\cdot)}(\outputRV)}
         {\density{\inputDistribution}(\inputRV)\density{\compoundChannel{\stateSpaceElement_0}(\inputRV,\cdot)}(\outputRV)}
\\
&=
\Expectation_{\inputOutputDistribution{\inputDistribution}{\compoundChannel{\stateSpaceElement}}}
\log\frac{\density{\inputOutputDistribution{\inputDistribution}{\compoundChannel{\stateSpaceElement}}}(\inputRV,\outputRV)}
         {\density{\inputOutputDistribution{\inputDistribution}{\compoundChannel{\stateSpaceElement_0}}}(\inputRV,\outputRV)}
\\
&=
\kldiv{\inputOutputDistribution{\inputDistribution}{\compoundChannel{\stateSpaceElement}}}{\inputOutputDistribution{\inputDistribution}{\compoundChannel{\stateSpaceElement_0}}}.
\end{align*}
From our arguments above, given any $\stateSpaceElement$, the distribution $\inputOutputDistribution{\inputDistribution}{\compoundChannel{\stateSpaceElement}}$ is multivariate Gaussian with positive definite covariance matrix, which implies that (\ref{eq:topological-approximation-kldiv}) is satisfied.
\end{proof}

\subsection{Coding Result}
\label{sec:compound-coding-result}
We use the same random codebook construction that was originally employed by Shannon~\cite{shannon1948mathematical}: Given a channel input alphabet $\inputAlphabet$, a distribution $\inputDistribution$ on $\inputAlphabet$, a block length $\blocklength$ and a rate $\codebookRate$, we define the \gls{codebookensemble} of code books as a random experiment in which $\exp(\blocklength\codebookRate)$ code words of length $\blocklength$ are drawn randomly and independently according to $\inputDistribution$ for each component of each code word.

\begin{theorem}
\label{theorem:compound}
Let $(\compoundChannel{\stateSpaceElement})_{\stateSpaceElement \in \stateSpace}$ be a compound channel with input alphabet $\inputAlphabet$ and output alphabet $\outputAlphabet$, and let $\inputDistribution$ be a probability distribution on $\inputAlphabet$ such that for every $\channelApproximationError > 0$, there is a $\channelApproximationNumber{\channelApproximationError}$ such that $(\compoundChannel{\stateSpaceElement})_{\stateSpaceElement \in \stateSpace}$ can be $(\channelApproximationError,\channelApproximationNumber{\channelApproximationError})$-approximated under $\inputDistribution$. Let
\begin{align}
\label{eq:compound-ratecondition}
0 < \codebookRate < \inf_{\stateSpaceElement \in \stateSpace} \information{\inputDistribution}{\compoundChannel{\stateSpaceElement}},
\end{align}
and let \gls{codebook} be a random codebook from the $(\inputDistribution, \blocklength, \codebookRate)$-ensemble. Define an encoder $\codewordIndex \mapsto \codeword{\codewordIndex}$. Then there is a decoder such that the error probability $\errorProb$ of the resulting compound channel code satisfies
\begin{align}
\label{eq:compound-finalerror}
\Expectation_\codebook(\errorProb) < \exp(-\blocklength\finalconst),
\end{align}
for some $\finalconst > 0$ and sufficiently large $\blocklength$.
\end{theorem}

\begin{proof}
We first pick parameters $\channelApproximationError$, $\typicalityParameter$, $\boundParamOne$ and $\boundParamTwo$ in sequence according to the following scheme, where (\ref{eq:compound-ratecondition}) and the previous choices ensure that these intervals are all nonempty.
\begin{align}
\label{eq:compound-pickapproximationerror}
\channelApproximationError
&\in
\left(0, \frac{\inf_{\stateSpaceElement \in \stateSpace} \information{\inputDistribution}{\compoundChannel{\stateSpaceElement}} - \codebookRate}{3}\right)
\\
\label{eq:compound-picktypicality}
\typicalityParameter
&\in
\left(2\channelApproximationError, \inf_{\stateSpaceElement \in \stateSpace} \information{\inputDistribution}{\compoundChannel{\stateSpaceElement}} - \codebookRate - \channelApproximationError\right)
\\
\label{eq:compound-pickbound1}
\boundParamOne &\in (\channelApproximationError, \typicalityParameter - \channelApproximationError)
\\
\label{eq:compound-pickbound2}
\boundParamTwo &\in (0, \typicalityParameter - \channelApproximationError - \boundParamOne)
\end{align}

Fix a sequence $(\approximateChannel{\channelApproximationError}{\channelApproximationIndex})_{\channelApproximationIndex=1}^{\channelApproximationNumber{\channelApproximationError}}$ which $(\channelApproximationError,\channelApproximationNumber{\channelApproximationError})$-approximates $(\compoundChannel{\stateSpaceElement})_{\stateSpaceElement \in \stateSpace}$.

We use a joint typicality decoder, i.e., if there is a unique $\codewordIndex$ such that
\[
\exists \channelApproximationIndex \in \{1, \dots, \channelApproximationNumber{\channelApproximationError}\}:~
\informationDensity{\inputDistribution}
                   {\approximateChannel{\channelApproximationError}
                                       {\channelApproximationIndex}
                   }
                   {\codeword{\codewordIndex}}
                   {\outputRV^\blocklength}
\geq
\blocklength(
  \information{\inputDistribution}
              {\approximateChannel{\channelApproximationError}
                                  {\channelApproximationIndex}
              }
  -
  \typicalityParameter
),
\]
the decoder declares that message $\codewordIndex$ has been sent; otherwise it declares an error (or that message $1$ has been sent).

We denote the transmitted message with $\messageRV$, the message declared by the decoder with $\hat{\messageRV}$ and define error events
\begin{align}
\errorevent &:=
\{ \messageRV \neq \hat{\messageRV} \}
\\
\errorevent_1 &:=
\label{eq:compound-typical-event}
\Big\{
  \forall \channelApproximationIndex \in \{1, \dots, \channelApproximationNumber{\channelApproximationError}\} ~
  \informationDensity{\inputDistribution}
                     {\approximateChannel{\channelApproximationError}
                                         {\channelApproximationIndex}
                     }
                     {\codeword{\messageRV}}
                     {\outputRV^\blocklength}
  <
  \blocklength(
    \information{\inputDistribution}
                {\approximateChannel{\channelApproximationError}
                                    {\channelApproximationIndex}
                }
    -
    \typicalityParameter
  )
\Big\}
\\
\label{eq:compound-atypical-event}
\errorevent_2 &:= 
\Big\{
  \exists \codewordIndex \neq \messageRV ~
  \exists \channelApproximationIndex \in \{1, \dots, \channelApproximationNumber{\channelApproximationError}\} ~
  \informationDensity{\inputDistribution}
                     {\approximateChannel{\channelApproximationError}
                                         {\channelApproximationIndex}
                     }
                     {\codeword{\codewordIndex}}
                     {\outputRV^\blocklength}
  \geq
  \blocklength(
    \information{\inputDistribution}
                {\approximateChannel{\channelApproximationError}
                                    {\channelApproximationIndex}
                }
    -
    \typicalityParameter
  )
\Big\}.
\end{align}
We note that $\errorevent \subseteq \errorevent_1 \cup \errorevent_2$ and consequently

\begin{align}
\label{eq:compound-typical-atypical-split}
\Probability(\errorevent) \leq \Probability(\errorevent_1) + \Probability(\errorevent_2).
\end{align}
So we can bound these two errors separately and then combine them.

We start with bounding the expectation of the first summand, using the definition (\ref{eq:compound-typical-event}) and $\codebook$, as well as an addition of zero. Pick $\channelApproximationIndex$ such that $\approximateChannel{\channelApproximationError}{\channelApproximationIndex}$ satisfies (\ref{eq:approximation-kldiv1}) -- (\ref{eq:approximation-information1}) with respect to the realization $\compoundChannel{\stateSpaceElement}$ of the compound channel. Then we have
\begin{align}
\nonumber
\Expectation_\codebook(\Probability(\errorevent_1))
\nonumber
&\leq
\Expectation_\codebook\left(
  \Probability\left(
    \informationDensity{\inputDistribution}
                      {\approximateChannel{\channelApproximationError}
                                          {\channelApproximationIndex}
                      }
                      {\codeword{\messageRV}}
                      {\outputRV^\blocklength}
    <
    \blocklength(
      \information{\inputDistribution}
                  {\approximateChannel{\channelApproximationError}
                                      {\channelApproximationIndex}
                  }
      -
      \typicalityParameter
    )
  \right)
\right)
\\
\nonumber
&=
\inputOutputDistribution{\inputDistribution}{\compoundChannel{\stateSpaceElement}}^\blocklength\left(
  \informationDensity{\inputDistribution}
                    {\approximateChannel{\channelApproximationError}
                                        {\channelApproximationIndex}
                    }
                    {\inputRV^\blocklength}
                    {\outputRV^\blocklength}
  <
  \blocklength(
    \information{\inputDistribution}
                {\approximateChannel{\channelApproximationError}
                                    {\channelApproximationIndex}
                }
    -
    \typicalityParameter
  )
\right)
\\
\label{eq:compound-typical-initial}
&=
\inputOutputDistribution{\inputDistribution}{\compoundChannel{\stateSpaceElement}}^\blocklength\Bigg(
  \sum\limits_{\blockIndex=1}^{\blocklength}
  \log\Bigg(
    \RNDerivative{\approximateChannel{\channelApproximationError}{\channelApproximationIndex}(\inputRV_\blockIndex, \cdot)}
                 {\marginalOutputDistribution{\inputDistribution}{\approximateChannel{\channelApproximationError}{\channelApproximationIndex}}}
    (\outputRV_\blockIndex)
  \Bigg)
  \blocklength(
    \information{\inputDistribution}
                {\approximateChannel{\channelApproximationError}
                                    {\channelApproximationIndex}
                }
    +
    \information{\inputDistribution}
                {\compoundChannel{\stateSpaceElement}}
    -
    \information{\inputDistribution}
                {\compoundChannel{\stateSpaceElement}}
    -
    \typicalityParameter
  )
\Bigg)
\end{align}

The Radon-Nikodym derivative can be split as
\begin{align}
\label{eq:compound-typical-rnderivative-split}
\RNDerivative{\approximateChannel{\channelApproximationError}{\channelApproximationIndex}(\inputRV_\blockIndex, \cdot)}
              {\marginalOutputDistribution{\inputDistribution}{\approximateChannel{\channelApproximationError}{\channelApproximationIndex}}}
=
\RNDerivative{\approximateChannel{\channelApproximationError}{\channelApproximationIndex}(\inputRV_\blockIndex, \cdot)}
             {\compoundChannel{\stateSpaceElement}(\inputRV_\blockIndex, \cdot)}
\cdot
\RNDerivative{\marginalOutputDistribution{\inputDistribution}{\compoundChannel{\stateSpaceElement}}}
             {\marginalOutputDistribution{\inputDistribution}{\approximateChannel{\channelApproximationError}{\channelApproximationIndex}}}
\cdot
\RNDerivative{\compoundChannel{\stateSpaceElement}(\inputRV_\blockIndex, \cdot)}
             {\marginalOutputDistribution{\inputDistribution}{\compoundChannel{\stateSpaceElement}}}.
\end{align}
This is possible because $\approximateChannel{\channelApproximationError}{\channelApproximationIndex}(\inputAlphabetElement, \cdot) \absolutelyContinuous \compoundChannel{\stateSpaceElement}(\inputAlphabetElement, \cdot)$ by (\ref{eq:approximation-absolute-continuity}), $\marginalOutputDistribution{\inputDistribution}{\compoundChannel{\stateSpaceElement}} \absolutelyContinuous \marginalOutputDistribution{\inputDistribution}{\approximateChannel{\channelApproximationError}{\channelApproximationIndex}}$ by (\ref{eq:approximation-kldiv1}) and the joint convexity of Kullback-Leibler divergence in its arguments, and $\compoundChannel{\stateSpaceElement}(\inputAlphabetElement, \cdot) \absolutelyContinuous \marginalOutputDistribution{\inputDistribution}{\compoundChannel{\stateSpaceElement}}$ for $\inputDistribution$-almost all $\inputAlphabetElement$ by the properties of the marginalization.

We next bound tail probabilities corresponding to the three factors in (\ref{eq:compound-typical-rnderivative-split}) separately, starting with the first. To this end, we introduce a number $\renyiParamOne > 1$ and argue, using Markov's inequality and the definition of Rényi divergence, that
\begin{align}
\nonumber
\inputOutputDistribution{\inputDistribution}{\compoundChannel{\stateSpaceElement}}^\blocklength\left(
  \sum\limits_{\blockIndex=1}^{\blocklength}
  \log\RNDerivative{\compoundChannel{\stateSpaceElement}(\inputRV_\blockIndex, \cdot)}
                   {\approximateChannel{\channelApproximationError}{\channelApproximationIndex}(\inputRV_\blockIndex, \cdot)}
  (\outputRV_\blockIndex)
  \geq
  \blocklength \boundParamOne
\right)
&=
\inputOutputDistribution{\inputDistribution}{\compoundChannel{\stateSpaceElement}}^\blocklength\Bigg(
  \exp\left(
    (\renyiParamOne-1)
    \sum\limits_{\blockIndex=1}^{\blocklength}
    \log\RNDerivative{\compoundChannel{\stateSpaceElement}(\inputRV_\blockIndex, \cdot)}
                    {\approximateChannel{\channelApproximationError}{\channelApproximationIndex}(\inputRV_\blockIndex, \cdot)}
    (\outputRV_\blockIndex)
  \right)
  \exp((\renyiParamOne-1)\blocklength \boundParamOne)
\Bigg)
\\
\nonumber
&\leq
\Expectation_{\inputOutputDistribution{\inputDistribution}{\compoundChannel{\stateSpaceElement}}^\blocklength}\left(
  \left(
    \prod\limits_{\blockIndex=1}^{\blocklength}
    \left(
    \RNDerivative{\compoundChannel{\stateSpaceElement}(\inputRV_\blockIndex, \cdot)}
                 {\approximateChannel{\channelApproximationError}{\channelApproximationIndex}(\inputRV_\blockIndex, \cdot)}
    (\outputRV_\blockIndex)
    \right)^{\renyiParamOne-1}
  \right)
\right)
\cdot \exp(-(\renyiParamOne-1)\blocklength \boundParamOne)
\\
\nonumber
&=
\exp\left(
  \sum_{\blockIndex=1}^\blocklength
    \log\left(
      \Expectation_{\inputOutputDistribution{\inputDistribution}{\compoundChannel{\stateSpaceElement}}^\blocklength}
        \left(
          \left(
            \RNDerivative{\compoundChannel{\stateSpaceElement}(\inputRV_\blockIndex, \cdot)}
                         {\approximateChannel{\channelApproximationError}{\channelApproximationIndex}(\inputRV_\blockIndex, \cdot)}
          \right)^{\renyiParamOne-1}
        \right)
    \right)
\right)
\cdot \exp(-(\renyiParamOne-1)\blocklength \boundParamOne)
\\
\label{eq:compound-typical-first-factor}
&=
\exp\Big(
  -
  (\renyiParamOne-1)
  \blocklength
  \cdot
  \left(
    \boundParamOne
    -
    \Expectation_\inputDistribution
      \renyidiv{\renyiParamOne}
               {\compoundChannel{\stateSpaceElement}(\inputRV, \cdot)}
               {\approximateChannel{\channelApproximationError}{\channelApproximationIndex}(\inputRV, \cdot)}
  \right)
\Big).
\end{align}

For the second factor, we argue in an analogous way, but using $\renyiParamTwo > 0$.

\begin{align}
\nonumber
\marginalOutputDistribution{\inputDistribution}{\compoundChannel{\stateSpaceElement}}^\blocklength\left(
  \sum\limits_{\blockIndex=1}^{\blocklength}
  \log\RNDerivative{\marginalOutputDistribution{\inputDistribution}{\approximateChannel{\channelApproximationError}{\channelApproximationIndex}}}
                   {\marginalOutputDistribution{\inputDistribution}{\compoundChannel{\stateSpaceElement}}}
  (\outputRV_\blockIndex)
  \geq
  \blocklength \boundParamTwo
\right)
&=
\marginalOutputDistribution{\inputDistribution}{\compoundChannel{\stateSpaceElement}}^\blocklength\Bigg(
  \exp\left(
    \renyiParamTwo
    \sum\limits_{\blockIndex=1}^{\blocklength}
    \log\RNDerivative{\marginalOutputDistribution{\inputDistribution}{\approximateChannel{\channelApproximationError}{\channelApproximationIndex}}}
                    {\marginalOutputDistribution{\inputDistribution}{\compoundChannel{\stateSpaceElement}}}
    (\outputRV_\blockIndex)
  \right)
  \exp(\renyiParamTwo \blocklength \boundParamTwo)
\Bigg)
\\
\nonumber
&\leq
\Expectation_{\marginalOutputDistribution{\inputDistribution}{\compoundChannel{\stateSpaceElement}}^\blocklength}\left(
  \prod\limits_{\blockIndex=1}^{\blocklength}
  \left(
    \RNDerivative{\marginalOutputDistribution{\inputDistribution}{\approximateChannel{\channelApproximationError}{\channelApproximationIndex}}}
                 {\marginalOutputDistribution{\inputDistribution}{\compoundChannel{\stateSpaceElement}}}
    (\outputRV_\blockIndex)
  \right)^\renyiParamTwo
\right)
\exp(-\renyiParamTwo \blocklength \boundParamTwo)
\\
\label{eq:compound-typical-second-factor}
&=
\exp\left(
  (\renyiParamTwo-1)
  \blocklength
  \renyidiv{\renyiParamTwo}
           {\marginalOutputDistribution{\inputDistribution}{\approximateChannel{\channelApproximationError}{\channelApproximationIndex}}}
           {\marginalOutputDistribution{\inputDistribution}{\compoundChannel{\stateSpaceElement}}}
  -\renyiParamTwo \blocklength \boundParamTwo
\right).
\end{align}

Finally, for the third factor, we use $\renyiParamThree < 1$.

\begin{align}
\nonumber
&\hphantom{{}={}}
\inputOutputDistribution{\inputDistribution}{\compoundChannel{\stateSpaceElement}}^\blocklength\Big(
  \informationDensity{\inputDistribution}
                    {\compoundChannel{\stateSpaceElement}}
                    {\inputRV^\blocklength}
                    {\outputRV^\blocklength}
  <
  \blocklength(
    \information{\inputDistribution}
                {\compoundChannel{\stateSpaceElement}}
    -
    \typicalityParameter
    +
    \boundParamOne
    +
    \boundParamTwo
    +
    \channelApproximationError
  )
\Big)
\\
\nonumber
&=
\inputOutputDistribution{\inputDistribution}{\compoundChannel{\stateSpaceElement}}^\blocklength\big(
  \exp\left(
    (\renyiParamThree-1)
    \informationDensity{\inputDistribution}
                      {\compoundChannel{\stateSpaceElement}}
                      {\inputRV^\blocklength}
                      {\outputRV^\blocklength}
  \right)
  >
  \exp\left(
    (\renyiParamThree-1)
    \blocklength(
      \information{\inputDistribution}
                  {\compoundChannel{\stateSpaceElement}}
      -
      \typicalityParameter
      +
      \boundParamOne
      +
      \boundParamTwo
      +
      \channelApproximationError
    )
  \right)
\big)
\\
\nonumber
&\leq
\Expectation_{\inputOutputDistribution{\inputDistribution}{\compoundChannel{\stateSpaceElement}}^\blocklength}\left(
  \prod\limits_{\blockIndex=1}^{\blocklength}
  \left(
    \RNDerivative{\compoundChannel{\stateSpaceElement}(\inputRV_\blockIndex,\cdot)}
                 {\marginalOutputDistribution{\inputDistribution}{\compoundChannel{\stateSpaceElement}}}
    (\outputRV_\blockIndex)
  \right)^{\renyiParamThree-1}
\right)
\cdot
\exp\left(
  -(\renyiParamThree-1)
  \blocklength(
    \information{\inputDistribution}
                {\compoundChannel{\stateSpaceElement}}
    -
    \typicalityParameter
    +
    \boundParamOne
    +
    \boundParamTwo
    +
    \channelApproximationError
  )
\right)
\\
\label{eq:compound-typical-third-factor}
&=
\exp\Big(
  -
  (1-\renyiParamThree)
  \blocklength
  \big(
    \renyidiv{\renyiParamThree}
             {\inputOutputDistribution{\inputDistribution}{\compoundChannel{\stateSpaceElement}}}
             {\inputDistribution \marginalOutputDistribution{\inputDistribution}{\compoundChannel{\stateSpaceElement}}}
    +
    \typicalityParameter
    -
    \information{\inputDistribution}
                {\compoundChannel{\stateSpaceElement}}
    -
    \boundParamOne
    -
    \boundParamTwo
    -
    \channelApproximationError
  \big)
\Big),
\end{align}

Clearly, by (\ref{eq:compound-typical-rnderivative-split}), the union bound and (\ref{eq:approximation-information1}), (\ref{eq:compound-typical-initial}) is upper bounded by the sum of (\ref{eq:compound-typical-first-factor}), (\ref{eq:compound-typical-second-factor}) and (\ref{eq:compound-typical-third-factor}). Next, we argue that these expressions all vanish exponentially with $\blocklength \rightarrow \infty$, using the continuity of Rényi divergence in the order which is shown in~\cite[Theorem 7]{vanerven2014renyi}.

From (\ref{eq:approximation-renyidiv}), the theorem of monotone convergence and (\ref{eq:approximation-kldiv1}), we can conclude that
\[
\lim_{\renyiParamOne \searrow 1}
  \Expectation_\inputDistribution
    \renyidiv{\renyiParamOne}
              {\compoundChannel{\stateSpaceElement}(\inputRV_\blockIndex, \cdot)}
              {\approximateChannel{\channelApproximationError}{\channelApproximationIndex}(\inputRV_\blockIndex, \cdot)}
=
\Expectation_\inputDistribution
  \kldiv{\compoundChannel{\stateSpaceElement}(\inputRV_\blockIndex, \cdot)}
           {\approximateChannel{\channelApproximationError}{\channelApproximationIndex}(\inputRV_\blockIndex, \cdot)}
\leq
\channelApproximationError,
\]
so, (\ref{eq:compound-pickbound1}) allows us to fix $\renyiParamOne$ at a value greater than $1$ such that
$
    \boundParamOne
    -
    \Expectation_\inputDistribution
      \renyidiv{\renyiParamOne}
               {\compoundChannel{\stateSpaceElement}(\inputRV_\blockIndex, \cdot)}
               {\approximateChannel{\channelApproximationError}{\channelApproximationIndex}(\inputRV_\blockIndex, \cdot)}
>
0
$
and hence, (\ref{eq:compound-typical-first-factor}) vanishes exponentially.

(\ref{eq:compound-typical-second-factor}) is true for all $\renyiParamTwo < 1$. Since the inequalities are not strict, we can take the limit $\renyiParamTwo \nearrow 1$ and argue that the statement is also valid for $\renyiParamTwo = 1$.

$
\renyidiv{\renyiParamThree}
          {\inputOutputDistribution{\inputDistribution}{\compoundChannel{\stateSpaceElement}}}
          {\inputDistribution \marginalOutputDistribution{\inputDistribution}{\compoundChannel{\stateSpaceElement}}}
$
converges to
$
    \information{\inputDistribution}
                {\compoundChannel{\stateSpaceElement}}
$
from below for $\renyiParamThree \nearrow 1$ and so (\ref{eq:compound-pickbound2}) allows us to fix $\renyiParamThree$ at a value less than $1$ such that
$
\renyidiv{\renyiParamThree}
          {\inputOutputDistribution{\inputDistribution}{\compoundChannel{\stateSpaceElement}}}
          {\inputDistribution \marginalOutputDistribution{\inputDistribution}{\compoundChannel{\stateSpaceElement}}}
+
\typicalityParameter
-
\information{\inputDistribution}
            {\compoundChannel{\stateSpaceElement}}
-
\boundParamOne
-
\boundParamTwo
-
\channelApproximationError
>
0
$
and therefore, (\ref{eq:compound-typical-third-factor}) also vanishes exponentially.

For the second summand in (\ref{eq:compound-typical-atypical-split}), we use the definition (\ref{eq:compound-atypical-event}) to argue that
$
\Expectation_\codebook(\Probability(\errorevent_2))
$
is upper bounded by

\begin{equation}
\label{eq:compound-atypical-apply-def}
\exp(\blocklength\codebookRate)
\sum\limits_{\channelApproximationIndex=1}^\channelApproximationNumber{\channelApproximationError}
  \inputDistribution^\blocklength \marginalOutputDistribution{\inputDistribution}{\compoundChannel{\stateSpaceElement}}^\blocklength\bigg(
    \informationDensity{\inputDistribution}{\approximateChannel{\channelApproximationError}{\channelApproximationIndex}}{\inputRV^\blocklength}{\outputRV^\blocklength}
    \geq
    \blocklength\left(\information{\inputDistribution}{\approximateChannel{\channelApproximationError}{\channelApproximationIndex}} - \typicalityParameter \right)
  \bigg).
\end{equation}

We define the indicator function
\[
\indicatorAbbrev(\inputAlphabetElement^\blocklength, \outputAlphabetElement^\blocklength)
:=
\begin{cases}
1,
&\informationDensity{\inputDistribution}{\approximateChannel{\channelApproximationError}{\channelApproximationIndex}}{\inputAlphabetElement^\blocklength}{\outputAlphabetElement^\blocklength}
    \geq
    \blocklength\left(\information{\inputDistribution}{\approximateChannel{\channelApproximationError}{\channelApproximationIndex}} - \typicalityParameter \right) \\
0, &\text{otherwise}.
\end{cases}
\]
Using the definition of information density for a change of measure and multiplying one, we rewrite the probability that appears in (\ref{eq:compound-atypical-apply-def}) as
\begin{align*}
&\hphantom{{}={}}
\inputDistribution^\blocklength \marginalOutputDistribution{\inputDistribution}{\compoundChannel{\stateSpaceElement}}^\blocklength\bigg(
  \informationDensity{\inputDistribution}{\approximateChannel{\channelApproximationError}{\channelApproximationIndex}}{\inputRV^\blocklength}{\outputRV^\blocklength}
  \geq
  \blocklength\left(\information{\inputDistribution}{\approximateChannel{\channelApproximationError}{\channelApproximationIndex}} - \typicalityParameter \right)
\bigg)
\\
&=
\hspace{-3pt}
\int\limits_{\inputAlphabet^\blocklength \times \outputAlphabet^\blocklength}
  \indicatorAbbrev(\inputAlphabetElement^\blocklength, \outputAlphabetElement^\blocklength)
  \cdot \inputDistribution^\blocklength \marginalOutputDistribution{\inputDistribution}{\compoundChannel{\stateSpaceElement}}^\blocklength(d\inputAlphabetElement^\blocklength,d\outputAlphabetElement^\blocklength)
\\
&=
\hspace{-3pt}
\int\limits_{\inputAlphabet^\blocklength \times \outputAlphabet^\blocklength}
  \exp\left(
    -
    \informationDensity{\inputDistribution}
                       {\compoundChannel{\stateSpaceElement}}
                       {\inputAlphabetElement^\blocklength}
                       {\outputAlphabetElement^\blocklength}
  \right)
  \indicatorAbbrev(\inputAlphabetElement^\blocklength, \outputAlphabetElement^\blocklength)
  \inputOutputDistribution{\inputDistribution}{\compoundChannel{\stateSpaceElement}}^\blocklength(d\inputAlphabetElement^\blocklength,d\outputAlphabetElement^\blocklength)
\\
&=
\hspace{-3pt}
\int\limits_{\inputAlphabet^\blocklength \times \outputAlphabet^\blocklength}
  \exp
    \big(-
    \informationDensity{\inputDistribution}
                       {\compoundChannel{\stateSpaceElement}}
                       {\inputAlphabetElement^\blocklength}
                       {\outputAlphabetElement^\blocklength}
    +
    \informationDensity{\inputDistribution}
                       {\approximateChannel{\channelApproximationError}{\channelApproximationIndex}}
                       {\inputAlphabetElement^\blocklength}
                       {\outputAlphabetElement^\blocklength}
    -
    \informationDensity{\inputDistribution}
                       {\approximateChannel{\channelApproximationError}{\channelApproximationIndex}}
                       {\inputAlphabetElement^\blocklength}
                       {\outputAlphabetElement^\blocklength}
  \big)
  \cdot \indicatorAbbrev(\inputAlphabetElement^\blocklength, \outputAlphabetElement^\blocklength)
  \inputOutputDistribution{\inputDistribution}{\compoundChannel{\stateSpaceElement}}^\blocklength(d\inputAlphabetElement^\blocklength,d\outputAlphabetElement^\blocklength)
\end{align*}

Because of the presence of the indicator, we can uniformly bound
\[
\informationDensity{\inputDistribution}{\approximateChannel{\channelApproximationError}{\channelApproximationIndex}}{\inputAlphabetElement^\blocklength}{\outputAlphabetElement^\blocklength}
\geq
\blocklength\left(\information{\inputDistribution}{\approximateChannel{\channelApproximationError}{\channelApproximationIndex}} - \typicalityParameter \right)
\]
and the indicator itself can be upper bounded by $1$. This yields
\begin{align*}
&\hphantom{{}={}}
\inputDistribution^\blocklength \marginalOutputDistribution{\inputDistribution}{\compoundChannel{\stateSpaceElement}}^\blocklength\bigg(
  \informationDensity{\inputDistribution}{\approximateChannel{\channelApproximationError}{\channelApproximationIndex}}{\inputRV^\blocklength}{\outputRV^\blocklength}
  \geq
  \blocklength\left(\information{\inputDistribution}{\approximateChannel{\channelApproximationError}{\channelApproximationIndex}} - \typicalityParameter \right)
\bigg)
\\
&\leq
\begin{multlined}[t]
  \exp\left(-\blocklength\left(\information{\inputDistribution}{\approximateChannel{\channelApproximationError}{\channelApproximationIndex}} - \typicalityParameter \right)\right)
  \\
  \int\limits_{\inputAlphabet^\blocklength \times \outputAlphabet^\blocklength}
  \begin{aligned}[t]
    &\exp\Big(
        -
        \informationDensity{\inputDistribution}
                          {\compoundChannel{\stateSpaceElement}}
                          {\inputAlphabetElement^\blocklength}
                          {\outputAlphabetElement^\blocklength}
        +
        \informationDensity{\inputDistribution}
                          {\approximateChannel{\channelApproximationError}{\channelApproximationIndex}}
                          {\inputAlphabetElement^\blocklength}
                          {\outputAlphabetElement^\blocklength}
    \Big)
  \\
  &\cdot \inputOutputDistribution{\inputDistribution}{\compoundChannel{\stateSpaceElement}}^\blocklength(d\inputAlphabetElement^\blocklength,d\outputAlphabetElement^\blocklength).
  \end{aligned}
\end{multlined}
\end{align*}
We expand the definition of information density and apply Fubini's Theorem to rewrite the integral as
\[
\int\limits_{\outputAlphabet^\blocklength}\left(
  \int\limits_{\inputAlphabet^\blocklength}
    \RNDerivative{\approximateChannel{\channelApproximationError}{\channelApproximationIndex}^\blocklength(\inputAlphabetElement^\blocklength, \cdot)}
                 {\marginalOutputDistribution{\inputDistribution}{\approximateChannel{\channelApproximationError}{\channelApproximationIndex}}^\blocklength}
    (\outputAlphabetElement^\blocklength)
    \inputDistribution^\blocklength(d\inputAlphabetElement^\blocklength)
  \right)
  \marginalOutputDistribution{\inputDistribution}{\compoundChannel{\stateSpaceElement}}^\blocklength(d\outputAlphabetElement^\blocklength)
\]
and observe that it equals $1$.

Combining with (\ref{eq:compound-atypical-apply-def}) and applying (\ref{eq:approximation-information2}), we obtain
\begin{align}
\nonumber
\Expectation_\codebook(\Probability(\errorevent_2))
&\leq
\exp(\blocklength\codebookRate)
\sum\limits_{\channelApproximationIndex=1}^\channelApproximationNumber{\channelApproximationError}
  \exp\left(-\blocklength\left(\information{\inputDistribution}{\approximateChannel{\channelApproximationError}{\channelApproximationIndex}} - \typicalityParameter \right)\right)
\\
\nonumber
&\leq
\exp(\blocklength\codebookRate)
\sum\limits_{\channelApproximationIndex=1}^\channelApproximationNumber{\channelApproximationError}
  \exp\left(-\blocklength\left(\inf_{\stateSpaceElement \in \stateSpace}\information{\inputDistribution}{\compoundChannel{\stateSpaceElement}} - \typicalityParameter  - \channelApproximationError \right)\right)
\\
&=
\label{eq:compound-atypical-final}
\exp\left(
  -
  \blocklength\left(
    \inf_{\stateSpaceElement \in \stateSpace}\information{\inputDistribution}{\compoundChannel{\stateSpaceElement}}
    -
    \typicalityParameter 
    -
    \codebookRate
    -
    \channelApproximationError
    -
    \frac{\log \channelApproximationNumber{\channelApproximationError}}{\blocklength}
  \right)
\right).
\end{align}

We observe that by (\ref{eq:compound-picktypicality}), 
$
\inf_{\stateSpaceElement \in \stateSpace}\information{\inputDistribution}{\compoundChannel{\stateSpaceElement}}
-
\typicalityParameter 
-
\codebookRate
-
\channelApproximationError
>
0
$.

Finally, we pick
\begin{align*}
\finalconst
\in
\Bigg(
  0,
  \min\bigg(
    &\begin{aligned}[t]
    &
    (\renyiParamOne-1)
     \cdot \left(
      \boundParamOne
      -
      \Expectation_\inputDistribution
        \renyidiv{\renyiParamOne}
                {\compoundChannel{\stateSpaceElement}(\inputRV_\blockIndex, \cdot)}
                {\approximateChannel{\channelApproximationError}{\channelApproximationIndex}(\inputRV_\blockIndex, \cdot)}
    \right),
    \\
    &\boundParamTwo
    ,
    \\
    &
    (1-\renyiParamThree)
    \big(
      \renyidiv{\renyiParamThree}
              {\inputOutputDistribution{\inputDistribution}{\compoundChannel{\stateSpaceElement}}}
              {\inputDistribution \marginalOutputDistribution{\inputDistribution}{\compoundChannel{\stateSpaceElement}}}
      +
      \typicalityParameter
      -
      \information{\inputDistribution}
                  {\compoundChannel{\stateSpaceElement}}
      -
      \boundParamOne
      -
      \boundParamTwo
      -
      \channelApproximationError
    \big)
    ,
    \end{aligned}
    \\
    &\inf_{\stateSpaceElement \in \stateSpace}\information{\inputDistribution}{\compoundChannel{\stateSpaceElement}}
    -
    \typicalityParameter 
    -
    \codebookRate
    -
    \channelApproximationError
  \bigg)
\Bigg).
\end{align*}
Since the exponent in (\ref{eq:compound-atypical-final}) is then negative for sufficiently large $\blocklength$, we can combine it with (\ref{eq:compound-typical-first-factor}), (\ref{eq:compound-typical-second-factor}) and (\ref{eq:compound-typical-third-factor}) to obtain (\ref{eq:compound-finalerror}).
\end{proof}

\subsection{Cost Constraint}
\label{sec:cost-constraint}
In this section, we use standard techniques to extend Theorem~\ref{theorem:compound} to the case of cost-constrained code books. We define an \emph{additive cost constraint} \gls{costConstraint} for an input alphabet $\inputAlphabet$  consisting of a function $\costFunction: \inputAlphabet \rightarrow [0,\infty)$ and a number $\costConstraint \in [0,\infty)$. Given any $\blocklength$, we say that $\inputAlphabetElement^\blocklength \in \inputAlphabet^\blocklength$ \emph{satisfies} the cost constraint if
$
\sum_{\blockIndex=1}^\codebookBlocklength
  \costFunction(\inputAlphabetElement_\blockIndex)
\leq
\codebookBlocklength
\costConstraint
$.

The specialization of this definition to a usual average power constraint would be to pick the square function as $\costFunction$ and the maximum admissible average power as $\costConstraint$.

As long as there is at least one $\inputAlphabetElement^\blocklength \in \inputAlphabet^\blocklength$ which satisfies the cost constraint $(\costFunction, \costConstraint)$, given any codebook $\codebook$ of block length $\blocklength$, we can define an associated \emph{cost-constrained codebook} $\codebook_{\costFunction, \costConstraint}$ which is generated from $\codebook$ by replacing all code words that do not satisfy the cost constraint with $\inputAlphabetElement^\blocklength$. Obviously, all code words in a cost-constrained codebook satisfy the cost constraint. We say that a cost constraint $(\costFunction,\costConstraint)$ is \emph{compatible} with an input distribution $\inputDistribution$ if for a random variable $\inputRV$ distributed according to $\inputDistribution$, $\costFunction(\inputRV)$ has a finite moment generating function in an interval containing $0$ in its interior and $\costConstraint > \Expectation_\inputDistribution \costFunction(\inputRV)$.

With these preliminary definitions, we can now state the compound channel coding result under an additive cost constraint.

\begin{cor}
\label{cor:compound}
In the setting of Theorem~\ref{theorem:compound}, and given an additive cost constraint $(\costFunction,\costConstraint)$ compatible with $\inputDistribution$, there are $\finalconstOne, \finalconstTwo > 0$ such that for sufficiently large $\blocklength$,
\begin{align}
\label{eq:cor-compound}
\Probability_{\codebook_{\costFunction,\costConstraint}}\big(\errorProb \geq \exp(-\blocklength \finalconstOne)\big) < \exp(-\blocklength \finalconstTwo).
\end{align}
\end{cor}

The approach used in the proof of Corollary~\ref{cor:compound} is similar to the one in~\cite[Section 3.3]{elgamal2011network}, but we include the adapted derivations in full here for the sake of self-containedness. Our approach is based on the idea that in probabilistic constructions, the union bound assures that even exponentially many constraints that are satisfied with a super-exponential error bound individually are simultaneously satisfied except for an error event of super-exponentially small probability. Such ideas have already been used in earlier works of information theory, such as~\cite{ahlswede1978elimination} and, in the context of Gaussian channels,~\cite{he2014mimo}. We begin with a series of preliminary lemmas and conclude the section with the proof of Corollary~\ref{cor:compound}.

\begin{lemma}
\label{lemma:lln-exponential}
Let $(\generalrvOne_\generalSummationIndex)_{\generalSummationIndex \geq 1}$ be a sequence of independent and identically distributed random variables such that the moment generating function $\momentGeneratingFunction(\chernoffParam) := \Expectation \exp(\chernoffParam \generalrvOne_1)$ exists on an interval containing $0$ in its interior. Let $\costConstraint > \Expectation \generalrvOne_1$. Then there exists $\finalconst > 0$ such that
\[
\Probability
\left(
  \sum\limits_{\generalSummationIndex=1}^\generalSummationBound \generalrvOne_\generalSummationIndex > \generalSummationBound \costConstraint
\right)
\leq
\exp(-\generalSummationBound\finalconst).
\]
\end{lemma}
\begin{proof}
We can without loss of generality assume that $\costConstraint=0$ and $\Expectation(\generalrvOne_1) < 0$, because otherwise we could consider the random variables $(\generalrvOne_\generalSummationIndex - \costConstraint)_{\generalSummationIndex \geq 1}$ instead.

Clearly, $\momentGeneratingFunction(0) = 1$ and $\momentGeneratingFunction'(0) = \Expectation(\generalrvOne_1) < 0$, so we can find some $\chernoffParam > 0$ sufficiently small such that $\momentGeneratingFunction(\chernoffParam) < 1$. With this choice of $\chernoffParam$, we can apply Markov's inequality and get
\begin{align*}
\Probability
\left(
  \sum\limits_{\generalSummationIndex=1}^\generalSummationBound \generalrvOne_\generalSummationIndex > 0
\right)
&=
\Probability
\left(
  \exp\left(
    \chernoffParam
    \sum\limits_{\generalSummationIndex=1}^\generalSummationBound \generalrvOne_\generalSummationIndex
  \right)
  >
  1
\right)
\\
&\leq
\Expectation
\left(
  \exp\left(
    \chernoffParam
    \sum\limits_{\generalSummationIndex=1}^\generalSummationBound \generalrvOne_\generalSummationIndex
  \right)
\right)
\\
&=
\momentGeneratingFunction(\chernoffParam)^\generalSummationBound
\end{align*}
so the lemma follows by choosing $\finalconst := -\log \momentGeneratingFunction(\chernoffParam)$.
\end{proof}

\begin{lemma}
\label{lemma:bernoullidoubleexp}
Let $\lemmaBadCodewords$ be a Bernoulli random variable with $\exp(\blocklength\codebookRate)$ trials and success probability $\lemmapvalue \leq \exp(-\blocklength\proofconstantOne)$ where $\proofconstantOne<\codebookRate/2$. Then there are $\finalconstOne, \finalconstTwo > 0$ such that for sufficiently large $\blocklength$,
\begin{align}
\label{eq:bernoullidoubleexp}
\Probability(\lemmaBadCodewords > \exp(\blocklength(\codebookRate-\finalconstOne)))
\leq
\exp(-\exp(\blocklength\finalconstTwo)).
\end{align}
\end{lemma}
\begin{proof}
We choose $\finalconstOne$, $\finalconstTwo$ and $\proofconstantTwo$ such that $0 < \finalconstOne < \proofconstantOne < \proofconstantTwo < \codebookRate/2$ and $\finalconstTwo < \codebookRate - 2\proofconstantTwo$. Then
\begin{align}
\Probability
\big(
  \lemmaBadCodewords
  >
  \exp(\codebookBlocklength(\codebookRate - \finalconstOne))
\big)
\nonumber
&=
\Probability
\big(
  \lemmaBadCodewords
  >
  \lemmapvalue
  \exp(\codebookBlocklength\codebookRate)
  +
  (\exp(-\codebookBlocklength\finalconstOne)-\lemmapvalue)
  \exp(\codebookBlocklength\codebookRate)
\big)
\nonumber
\\
&\leq
\Probability
\big(
  \lemmaBadCodewords
  >
  \Expectation \lemmaBadCodewords
  +
  (\exp(-\codebookBlocklength\finalconstOne)-\exp(-\codebookBlocklength\proofconstantOne))
  \exp(\codebookBlocklength\codebookRate)
\big)
\nonumber
\\
&\leq
\Probability
\big(
  \lemmaBadCodewords
  >
  \Expectation \lemmaBadCodewords
  +
  \exp(-\codebookBlocklength\proofconstantTwo)
  \exp(\codebookBlocklength\codebookRate)
\big)
\label{proof:cost-badcodewords-proofconstantTwo}
\\
&\leq
\exp\left(
  -2 \frac{\big(\exp(-\codebookBlocklength\proofconstantTwo)\big)^2\big(\exp(\codebookBlocklength\codebookRate)\big)^2}
          {\exp(\blocklength \codebookRate)}
\right)
\label{proof:cost-badcodewords-chernoff-hoeffding}
\\
&=
\exp\left(
  -2 \exp(\codebookBlocklength(\codebookRate - 2\proofconstantTwo))
\right)
\\
&\leq
\exp(-\exp(\codebookBlocklength\finalconstTwo)),
\label{proof:cost-badcodewords-finalconstTwo}
\end{align}
where (\ref{proof:cost-badcodewords-chernoff-hoeffding}) follows by the Chernoff-Hoeffding bound as stated for instance in~\cite[Theorem 1.1, eq. (1.6)]{dubhashi2009concentration}.
\end{proof}

\begin{lemma}
\label{lemma:cost-badcodewords}
Let $\inputDistribution$ be a probability distribution on $\inputAlphabet$. Assume moreover that $\costFunction(\inputRV)$ has a moment generating function defined on an interval containing $0$ in its interior and that $\costConstraint > \Expectation_{\inputDistribution} \costFunction(\inputRV)$. Denote the number of bad code words in $\codebook$ with
\[
\lemmaBadCodewords
:=
\sum\limits_{\codewordIndex=1}^{\exp(\codebookBlocklength\codebookRate)}
  \indicator{
    \sum_{\blockIndex=1}^\codebookBlocklength
      \costFunction(\codeword{\codewordIndex}(\blockIndex))
      >
      \codebookBlocklength\costConstraint
  }.
\]
Then there are $\finalconstOne, \finalconstTwo > 0$ such that
\begin{align}
\label{eq:lln-exponential}
\Probability_\codebook
\left(
  \lemmaBadCodewords
  >
  \exp(\codebookBlocklength(\codebookRate-\finalconstOne))
\right)
\leq
\exp(-\exp(\codebookBlocklength\finalconstTwo)).
\end{align}
\end{lemma}

\begin{proof}
Since the code word components are independently and identically distributed, we can apply Lemma~\ref{lemma:lln-exponential} and obtain an arbitrarily small $\proofconstantOne > 0$ such that for all $\codewordIndex$,
\[
\lemmapvalue
:=
\Probability_\codebook
\left(
  \sum_{\blockIndex=1}^\codebookBlocklength
    \costFunction(\codeword{\codewordIndex}(\blockIndex))
  >
  \codebookBlocklength\costConstraint
\right)
\leq
\exp(-\codebookBlocklength\proofconstantOne).
\]
So since the code words are independent, $\lemmaBadCodewords$ is a Bernoulli variable with $\exp(\codebookBlocklength\codebookRate)$ trials and success probability $\lemmapvalue$, and an application of Lemma~\ref{lemma:bernoullidoubleexp} proves the conclusion.
\end{proof}

\begin{proof}[Proof of Corollary~\ref{cor:compound}]

Assume throughout the proof that $\blocklength$ is sufficiently large. By Lemma~\ref{lemma:cost-badcodewords}, we have $\hat{\finalconstOne},\hat{\finalconstTwo} \in (0,\infty)$ with
\begin{align}
\label{eq:cor-compound-badcodewords}
\Probability_\codebook(\hat{\errorevent})
\leq
\exp(-\exp(\blocklength\hat{\finalconstTwo})),
\end{align}
where
\[
\hat{\errorevent}
:= 
\{\Probability_\messageRV(\codeword{\messageRV} \neq \codewordsub{\costFunction, \costConstraint}{\messageRV}) > \exp(-\blocklength\hat{\finalconstOne}) \}.
\]
We denote the error of $\codebook$ with $\errorProb_\codebook$ and the error of $\codebook_{\costFunction,\costConstraint}$ with $\errorProb_{\codebook_{\costFunction,\costConstraint}}$. By Theorem~\ref{theorem:compound} and Markov's inequality, we have, for some $\hat{\finalconst} \in (0,\infty)$ given by the theorem and with choices $\tilde{\finalconstOne} \in (0,\min(\hat{\finalconst},\hat{\finalconstOne})), \tilde{\finalconstTwo} \in (0,\hat{\finalconst}-\tilde{\finalconstOne})$,
\begin{align}
\nonumber
\Probability_\codebook(\errorProb_\codebook \geq \exp(-\blocklength\tilde{\finalconstOne}))
&\leq
\Expectation_\codebook \errorProb_\codebook \exp(\blocklength\tilde{\finalconstOne})
\\
\nonumber
&\leq
\exp(-\blocklength(\hat{\finalconst} - \tilde{\finalconstOne}))
\\
\label{eq:cor-compound-noconstraint}
&\leq
\exp(-\blocklength\tilde{\finalconstTwo}).
\end{align}

Conditioned on the complement of $\hat{\errorevent}$, we have
\begin{align}
\nonumber
\errorProb_{\codebook_{\costFunction,\costConstraint}}
\overset{(a)}&{=}
\sup_{\stateSpaceElement \in \stateSpace}
  \Expectation_\messageRV\Big(
    \Probability_\stateSpaceElement\big(
      \messageRV \neq \genericDecoder(\outputRV^\blocklength)
      |
      \inputRV^\blocklength = \codebook_{\costFunction,\costConstraint}(\messageRV)
    \big)
  \Big)
\\
\nonumber
&=
\sup_{\stateSpaceElement \in \stateSpace}
  \sum_{\codewordIndex=1}^{\exp(\blocklength\codebookRate)}
    \exp(-\blocklength\codebookRate)
    \Probability_\stateSpaceElement\big(
      \codewordIndex \neq \genericDecoder(\outputRV^\blocklength)
      |
      \inputRV^\blocklength = \codebook_{\costFunction,\costConstraint}(\codewordIndex)
    \big)
\\
\nonumber
\overset{(b)}&{\leq}
\sup_{\stateSpaceElement \in \stateSpace}
  \hspace{-5pt}
  \sum_{\substack{\codewordIndex=1 \\ \codebook_{\costFunction,\costConstraint}(\codewordIndex)=\codebook(\codewordIndex)}}^{\exp(\blocklength\codebookRate)}
  \hspace{-18pt}
    \exp(-\blocklength\codebookRate)
    \Probability_\stateSpaceElement\big(
      \codewordIndex \neq \genericDecoder(\outputRV^\blocklength)
      |
      \inputRV^\blocklength = \codebook_{\costFunction,\costConstraint}(\codewordIndex)
    \big)
  +
  \sum_{\substack{\codewordIndex=1 \\ \codebook_{\costFunction,\costConstraint}(\codewordIndex)\neq\codebook(\codewordIndex)}}^{\exp(\blocklength\codebookRate)}
    \exp(-\blocklength\codebookRate)
\\
\label{eq:cor-compound-conditioned}
\overset{(a)}&{\leq}
\errorProb_{\codebook}
+
\exp(-\blocklength\hat{\finalconstOne}),
\end{align}
where the steps marked with (a) are by the definition of compound coding error, and (b) is by upper bounding some of the probabilities in the sum with $1$. We can now choose $\finalconstOne \in (0, \tilde{\finalconstOne})$ and obtain
\begin{align*}
\Probability_\codebook\big(
  \errorProb_{\codebook_{\costFunction,\costConstraint}}
  \geq
  \exp(-\blocklength\finalconstOne)
\big)
\overset{(a)}&{\leq}
\Probability_\codebook\big(
  \errorProb_{\codebook_{\costFunction,\costConstraint}}
  \geq
  \exp(-\blocklength\finalconstOne)
  |
  \neg \hat{\errorevent}
\big)
+
\Probability_\codebook(\hat{\errorevent})
\\
\overset{(\ref{eq:cor-compound-conditioned})}&{\leq}
\Probability_\codebook\big(
  \errorProb_{\codebook}
  +
  \exp(-\blocklength\hat{\finalconstOne})
  \geq
  \exp(-\blocklength\finalconstOne)
  |
  \neg \hat{\errorevent}
\big)
+
\Probability_\codebook(\hat{\errorevent})
\\
\overset{(a)}&{\leq}
\frac{
  \Probability_\codebook\big(
    \errorProb_{\codebook}
    \geq
    \exp(-\blocklength\finalconstOne)
    -
    \exp(-\blocklength\hat{\finalconstOne})
  \big)
}{1-\Probability_\codebook(\hat{\errorevent})}
+
\Probability_\codebook(\hat{\errorevent})
\\
\overset{(b)}&{\leq}
\frac{
  \Probability_\codebook\big(
    \errorProb_{\codebook}
    \geq
    \exp(-\blocklength\tilde{\finalconstOne})
  \big)
}{1-\Probability_\codebook(\hat{\errorevent})}
+
\Probability_\codebook(\hat{\errorevent})
\\
\overset{(\ref{eq:cor-compound-badcodewords}),(\ref{eq:cor-compound-noconstraint})}&{\leq}
\frac{
  \exp(-\blocklength\tilde{\finalconstTwo})
}{1-\exp(-\exp(\blocklength\hat{\finalconstTwo}))}
+
\exp(-\exp(\blocklength\hat{\finalconstTwo}))
\\
\overset{(c)}&{\leq}
\exp(-\blocklength\finalconstTwo),
\end{align*}
where the steps marked with (a) are by the law of total probability, step (b) is by the choices of $\finalconstOne,\tilde{\finalconstOne}$, and step (c) is valid for any choice of $\finalconstTwo \in (0,\tilde{\finalconstTwo})$.
\end{proof}

\section{Jamming Strategies Induced by Random Code Books}
\label{sec:general-channel}
In this section, we leverage the results of Section~\ref{sec:compound} in conjunction with a known channel resolvability result to establish the main technical contribution that goes into the proof of Theorem~\ref{theorem:awgn-scheme-real}. The results and arguments in this section (except for the proof of Theorem~\ref{theorem:awgn-scheme-real}) are not specific to \gls{awgn} channels. In this section, we therefore use the system model described in Section~\ref{sec:dfaj} without the specializations from Section~\ref{sec:AWGN}. We fix an arbitrary admissible \gls{dfa} scheme as defined in Section~\ref{sec:dfaj}. Such a scheme will induce effective channels for $\bob$ and $\eve$ as outlined in Fig.~\ref{fig:systemmodel}. We denote the legitimate user's effective channel, which is a stochastic kernel mapping from $\analogMessageAlphabet{1} \times \dots \times \analogMessageAlphabet{\aliceNum} \times \inputAlphabet$ to $\outputAlphabet$, by \gls{effectiveChannelBob} and the eavesdropper's effective channel, which is a stochastic kernel mapping from $\analogMessageAlphabet{1} \times \dots \times \analogMessageAlphabet{\aliceNum} \times \inputAlphabet$ to $\eveAlphabet$, by \gls{effectiveChannelEve}.

In this section, we analyze jamming strategies that are induced by a codebook in the following sense: The jammer draws a code word index $\messageRV$ uniformly at random and transmits $\codeword{\messageRV}$, the code word in $\codebook$ indexed by $\messageRV$. Therefore, the number of code words in the codebook controls the amount of randomness contained in the jamming signal. We use the same random ensemble of code books that is defined at the beginning of Section~\ref{sec:compound-coding-result}.

With these concepts and notations defined, we are ready to state the main result of this section, which gives sufficient conditions for the existence of a jamming scheme that can simultaneously ensure that the legitimate receiver is able to reconstruct the full jamming signal and limit the usefulness of the eavesdropper's received signal.

\begin{theorem}
\label{theorem:jammersystem}
Let $\inputDistribution$ be a jammer input distribution.
Suppose that for every $\channelApproximationError > 0$, there is some $\channelApproximationNumber{\channelApproximationError}$ such that the compound channel $(\compoundChannel{\stateSpaceElement})_{\stateSpaceElement \in \stateSpace}$ defined by $\stateSpace := \stateSpace_1 \times \ldots \times \stateSpace_\aliceNum$ and $\compoundChannel{(\analogMessageAlphabetElement{1}, \dots, \analogMessageAlphabetElement{\aliceNum})} := \effectiveChannelBob(\analogMessageAlphabetElement{1}, \dots, \analogMessageAlphabetElement{\aliceNum}, \cdot, \cdot)$ can be $(\channelApproximationError,\channelApproximationNumber{\channelApproximationError})$-approximated under $\inputDistribution$. Suppose further that for all $\analogMessageAlphabetElement{1} \in \analogMessageAlphabet{1}, \dots, \analogMessageAlphabetElement{\aliceNum} \in \analogMessageAlphabet{\aliceNum}$, the moment-generating function
\[
\Expectation \exp(\generalReal \cdot \informationDensity{\inputDistribution}{\effectiveChannelEve(\analogMessageAlphabetElement{1}, \dots, \analogMessageAlphabetElement{\aliceNum}, \cdot, \cdot)}{\inputRV}{\eveOutputRV})
\]
of the information density exists and is finite at some point $\generalReal > 0$.
Let ($\costFunction$, $\costConstraint$) be an additive cost constraint compatible with $\inputDistribution$, and let $\codebook$ be a random codebook from the $(\inputDistribution, \blocklength, \codebookRate)$-ensemble. Let $\codebookRate \in (0,\infty)$ such that
\begin{equation}
\label{eq:jammersystem-rateconstraint}
\sup_{\analogMessageAlphabetElement{1} \in \analogMessageAlphabet{1}, \dots, \analogMessageAlphabetElement{\aliceNum} \in \analogMessageAlphabet{\aliceNum}}
\information{\inputDistribution}{\effectiveChannelEve(\analogMessageAlphabetElement{1}, \dots, \analogMessageAlphabetElement{\aliceNum}, \cdot, \cdot)}
<
\codebookRate
<
\inf_{\analogMessageAlphabetElement{1} \in \analogMessageAlphabet{1}, \dots, \analogMessageAlphabetElement{\aliceNum} \in \analogMessageAlphabet{\aliceNum}}
\information{\inputDistribution}{\effectiveChannelBob(\analogMessageAlphabetElement{1}, \dots, \analogMessageAlphabetElement{\aliceNum}, \cdot, \cdot)}.
\end{equation}
Then there are numbers $\finalconstOne, \finalconstTwo, \finalconstThree, \finalconstFour > 0$ such that for sufficiently large $\blocklength$,
\begin{equation}
\label{eq:jammersystem-totvar}
\Probability_\codebook\bigg(
  \totalvariationlr{
    \eveMarginalOutputDistribution{\effectiveChannelEve^\blocklength(\analogMessageAlphabetElement{1}, \dots, \analogMessageAlphabetElement{\aliceNum}, \cdot, \cdot), \codebook_{\costFunction, \costConstraint}}
    -
    \marginalOutputDistribution{\inputDistribution}{\effectiveChannelEve(\analogMessageAlphabetElement{1}, \dots, \analogMessageAlphabetElement{\aliceNum}, \cdot, \cdot)}^\blocklength
  }
  \geq
  \exp(-\blocklength\finalconstOne)
\bigg)
< \exp(-\exp(\blocklength\finalconstTwo)),
\end{equation}
where \gls{eveMarginalOutputDistribution} denotes output of a channel $\channel^\blocklength$ given that a uniformly random code word from the codebook $\codebook$ is transmitted, and
\begin{equation}
\label{eq:jammersystem-reconstruction}
\Probability_\codebook\left(\errorevent\right)
<
\exp(-\blocklength\finalconstFour),
\end{equation}
where $\errorevent$ is the event that the jamming strategy induced by $\codebook_{\costFunction, \costConstraint}$ does not allow reconstruction of the jamming signal with error at most $\exp(-\blocklength \finalconstThree)$.
\end{theorem}

\begin{remark}
\label{remark:jammersystem-totvar}
The bound (\ref{eq:jammersystem-totvar}) compares two probability distributions. The first one,
$
\eveMarginalOutputDistribution{\effectiveChannelEve^\blocklength(\analogMessageAlphabetElement{1}, \dots, \analogMessageAlphabetElement{\aliceNum}, \cdot, \cdot), \codebook_{\costFunction, \costConstraint}},
$
is the distribution the eavesdropper observes if the jamming strategy follows the approach we propose in this paper. The second one,
$
\marginalOutputDistribution{\inputDistribution}{\effectiveChannelEve(\analogMessageAlphabetElement{1}, \dots, \analogMessageAlphabetElement{\aliceNum}, \cdot, \cdot)}^\blocklength,
$
is the distribution the eavesdropper observes if the jammer transmits white noise. In the sense made explicit in Theorem~\ref{theorem:jammersystem}, these two cases are ``almost'' the same. This comparison is exploited in the proof of Theorem~\ref{theorem:awgn-scheme-real} to prove the \gls{mse}-security guarantee.
\end{remark}

In order to prove this theorem, we decompose the system depicted in Fig.~\ref{fig:systemmodel} into smaller (and more easily analyzed) subsystems by considering only a subset of the depicted terminals at a time.

\paragraph{Considering the terminals $\alice{1}, \dots, \alice{\aliceNum}, \bob$}
\label{par:approximation}
This is the \gls{dfa} system model. This part of the system consists of transmitters $(\alice{\aliceIndex})_{\aliceIndex=1}^\aliceNum$ each of which holds a value $\analogMessageAlphabetElement{\aliceIndex} \in \analogMessageAlphabet{\aliceIndex}$ and a receiver $\bob$ which has the objective of estimating $\objectiveFunction(\analogMessageAlphabetElement{1}, \dots, \analogMessageAlphabetElement{\aliceNum})$. To this end, each transmitter $\alice{\aliceIndex}$ passes $\analogMessageAlphabetElement{\aliceIndex}$ through a pre-processor $\alicePreproc{\aliceIndex}$ independently $\blocklength$ times yielding a sequence $\processedAnalogMessageRV{\aliceIndex}^\blocklength$ of channel inputs. These are transmitted through $\blocklength$ independent uses of the channel, generating a sequence $\outputRV^\blocklength$ of channel outputs. The receiver passes this sequence through a post-processor $\bobPostproc^\blocklength$ which generates an approximation $\objectiveEstimator$ of $\objectiveFunction(\analogMessageAlphabetElement{1}, \dots, \analogMessageAlphabetElement{\aliceNum})$. As mentioned, the design of the pre- and post-processors depends heavily on the channel model and a particular class of functions $\objectiveFunction$. The idea is that the pre-processors, the channel and the post-processor work together to mimic the function $\objectiveFunction$, and any approach following this idea will be highly dependent on the particular structure of the channel and $\objectiveFunction$. In Theorem~\ref{theorem:jammersystem}, it is assumed that such a system is already in place and an augmentation is proposed which makes it more secure. A property of the system described in Section~\ref{sec:dfaj} necessary for our purposes and heavily exploited in this work is that the pre-processing is i.i.d., i.e., each pre-processor $\alicePreproc{\aliceIndex}$ is a stochastic kernel mapping from $\analogMessageAlphabet{\aliceIndex}$ to $\inputAlphabet_\aliceIndex$ and an $\blocklength$-fold product $\alicePreproc{\aliceIndex}^\blocklength$ of it is used to generate the channel input sequence.

\paragraph{Considering the terminals $\alice{1}, \dots, \alice{\aliceNum}, \jammer, \eve$}
\label{par:eavesdropping}
In this setting, we assume that the transmitters $\alice{1}, \dots, \alice{\aliceNum}$ run a scheme of the kind described under~\ref{par:approximation}). Instead of the legitimate receiver, there is now an eavesdropper $\eve$. The objective is then to limit the usefulness of the eavesdropper's received signal $\eveOutputRV^\blocklength$. To this end, we add a friendly jammer $\jammer$ to the system which transmits, according to a certain strategy, a word $\inputRV^\blocklength$. In this work, any jamming strategy we consider is induced by a codebook $\codebook$ of words of length $\blocklength$ through the rule that the jammer chooses an element of the codebook uniformly at random and transmits it. We use existing results on \emph{channel resolvability} to derive a bound on the usefulness of the signal $\eveOutputRV^\blocklength$ received at $\eve$.

\paragraph{Considering the terminals $\alice{1}, \dots, \alice{\aliceNum}, \jammer, \bob$}
\label{par:jammer-canceling}
This is the setting from~\ref{par:approximation}) with an additional transmitter $\jammer$. Here we assume that $\jammer$ uses a jamming strategy induced by a codebook $\codebook$ as described under~\ref{par:eavesdropping}) and use Theorem~\ref{theorem:compound} on compound channel coding to argue that for suitable choices of $\codebook$, $\bob$ is able to fully reconstruct the jamming signal $\inputRV^\blocklength$. This enables $\bob$ to perform a cancellation of the jamming signal before it applies the post-processor $\bobPostproc^\blocklength$ it would use in setting~\ref{par:approximation}). How this cancellation works depends on the particularities of the channel considered, but if, e.g., the jamming signal is simply added to the channel output as in the \gls{awgn} example in Section~\ref{sec:AWGN}, it is possible to cancel it entirely by subtracting it from the received signal. So in this case, the post-processor would consist of a reconstruction of the jamming signal, the subtraction of this signal from the received one and a post-processing step identical to that from~\ref{par:approximation}).

\paragraph{Combining settings~\ref{par:eavesdropping}) and~\ref{par:jammer-canceling})}
The goal here is to argue the existence of a codebook $\codebook$ which achieves both of the objectives described under~\ref{par:eavesdropping}) and~\ref{par:jammer-canceling}). It will turn out that this can be achieved by a standard random codebook construction.

Theorem~\ref{theorem:jammersystem} formulates conditions under which there are code books in the $(\inputDistribution, \blocklength, \codebookRate)$-ensemble of which the $(\costFunction, \costConstraint)$-cost constrained versions simultaneously achieve the goals set forth under~\ref{par:eavesdropping}) and~\ref{par:jammer-canceling}).

As a technical ingredient for our proof, we recall a result on channel resolvability from~\cite{frey2018resolvability} that will be applied in order to guarantee the virtual indistinguishability of the jamming signal from white noise for the eavesdropper.

\begin{theorem}{\cite{frey2018resolvability}}
\label{theorem:resolvability}
Given a channel $\channel$ from $\inputAlphabet$ to $\outputAlphabet$, an input distribution $\inputDistribution$ such that the moment-generating function $\Expectation_{\inputOutputDistribution{\inputDistribution}{\channel}} \exp(\generalReal \cdot \informationDensity{\inputDistribution}{\channel}{\inputRV}{\outputRV})$ of the information density exists and is finite for some $\generalReal > 0$, and $\codebookRate > \information{\inputDistribution}{\channel}$, there exist $\finalconstOne > 0$ and $\finalconstTwo > 0$ such that for large enough block lengths $\codebookBlocklength$, the $(\inputDistribution, \blocklength, \codebookRate)$-ensemble satisfies
\begin{equation}
\label{eq:resolvability-theorem}
\Probability_\codebook \left(
  \totalvariation{
    \eveMarginalOutputDistribution{\channel^\blocklength,\codebook} - \inputOutputDistribution{\inputDistribution}{\channel}^\blocklength
  }
  >
  \exp(-\finalconstOne\blocklength)
\right)
\leq
\exp\left(-\exp\left(\finalconstTwo\blocklength\right)\right),
\end{equation}
where $\eveMarginalOutputDistribution{\channel^\blocklength,\codebook}$ is the output distribution of channel $\channel$ given that a uniformly random code word from $\codebook$ is transmitted.
\end{theorem}

Similarly as with the compound channel coding theorem, we can use known methods to incorporate an additive cost constraint and argue the following corollary.

\begin{cor}
\label{cor:costconstraint}
Let $\inputDistribution$ be an input distribution on $\inputAlphabet$ and $(\costFunction, \costConstraint)$ an additive cost constraint compatible with $\inputDistribution$. Then the statement of Theorem \ref{theorem:resolvability} is valid even if the codebook $\codebook$ is replaced with its associated cost-constrained version $\codebook_{\costFunction, \costConstraint}$.
\end{cor}

\begin{proof}
By Lemma~\ref{lemma:cost-badcodewords}, we  pick $\hat{\finalconstOne}, \hat{\finalconstTwo}$ satisfying (\ref{eq:lln-exponential}) and by Theorem~\ref{theorem:resolvability}, we  pick $\tilde{\finalconstOne}, \tilde{\finalconstTwo}$ satisfying (\ref{eq:resolvability-theorem}).

We use the observation that $\lemmaBadCodewords \leq \exp((\codebookRate-\hat{\finalconstOne})\blocklength)$ implies
\begin{align}
\label{eq:cor-costconstraint-badcodewords}
\totalvariation{
  \eveMarginalOutputDistribution{\channel^\blocklength, \codebook_{\costFunction, \costConstraint}} - \eveMarginalOutputDistribution{\channel^\blocklength, \codebook}
}
\leq
\frac{\lemmaBadCodewords}
     {\exp(\blocklength\codebookRate)}
\leq
\exp(-\hat{\finalconstOne}\blocklength)
\end{align}
and observe that, as long as $\finalconstOne < \hat{\finalconstOne}, \tilde{\finalconstOne}$ and $\finalconstTwo < \hat{\finalconstTwo}, \tilde{\finalconstTwo}$ and $\blocklength$ is sufficiently large,
\begin{align*}
&\hphantom{{}={}}
\Probability_{\codebook_{\costFunction, \costConstraint}} \left(
  \totalvariation{
    \eveMarginalOutputDistribution{\channel^\blocklength, \codebook_{\costFunction, \costConstraint}} - \inputOutputDistribution{\inputDistribution}{\channel}^\blocklength
  }
  >
  \exp(-\finalconstOne\blocklength)
\right)
\\
\overset{(a)}&{\leq}
\Probability_{\codebook} \left(
  \totalvariation{
    \eveMarginalOutputDistribution{\channel^\blocklength, \codebook_{\costFunction, \costConstraint}} - \eveMarginalOutputDistribution{\channel^\blocklength, \codebook}
  }
  +
  \totalvariation{
    \eveMarginalOutputDistribution{\channel^\blocklength, \codebook} - \inputOutputDistribution{\inputDistribution}{\channel}^\blocklength
  }
  >
  \exp(-\finalconstOne\blocklength)
\right)
\\
\overset{(b)}&{\leq}
\Probability_{\codebook} \left(
  \totalvariation{
    \eveMarginalOutputDistribution{\channel^\blocklength, \codebook_{\costFunction, \costConstraint}} - \eveMarginalOutputDistribution{\channel^\blocklength, \codebook}
  }
  >
  \exp(-\hat{\finalconstOne}\blocklength)
\right)
+
\Probability_{\codebook} \left(
  \totalvariation{
    \eveMarginalOutputDistribution{\channel^\blocklength, \codebook} - \inputOutputDistribution{\inputDistribution}{\channel}^\blocklength
  }
  >
  \exp(-\tilde{\finalconstOne}\blocklength)
\right)
\\
\overset{(c)}&{<}
\exp(-\exp(\hat{\finalconstTwo}\blocklength)) + \exp(-\exp(\tilde{\finalconstTwo}\blocklength))
\\
\overset{(d)}&{\leq}
\exp(-\exp(\finalconstTwo\blocklength)),
\end{align*}
where (a) is by the triangle inequality, (b) is by the union bound and the choice of $\finalconstOne$, (c) is due to (\ref{eq:lln-exponential}), (\ref{eq:cor-costconstraint-badcodewords}) and (\ref{eq:resolvability-theorem}), and (d) is by the choice of $\finalconstTwo$.
\end{proof}

Given the previous observations, the proof of the main result of this section is now straightforward.

\begin{proof}[Proof of Theorem~\ref{theorem:jammersystem}]
An application of Corollary~\ref{cor:compound} yields (\ref{eq:jammersystem-reconstruction}), and (\ref{eq:jammersystem-totvar}) follows from Corollary~\ref{cor:costconstraint}.
\end{proof}

We can now put everything together and prove the main theorem of this paper.
\begin{proof}[Proof of Theorem~\ref{theorem:awgn-scheme-real}]
For the pre-processing at the transmitters, we use the same scheme as in the proof of Lemma~\ref{lemma:awgn-scheme-ideal} and begin by verifying that the resulting effective channels $\effectiveChannelBob$ and $\effectiveChannelEve$ with the input distribution $\inputDistribution$ chosen to be Gaussian with mean $0$ and variance $\jammerPowerConstraint$ satisfy the assumptions of Theorem~\ref{theorem:jammersystem}. Since the defined compound channel is a class of Gaussian channels with different means taking values in the compact set $[-1,1]$, the approximability of the channel is an immediate consequence of Theorem~\ref{theorem:gaussian-compound-approximation}. The finiteness of the moment-generating function of the information density can be seen by straightforward applications of the definitions of information density and Rényi divergence:
\begin{align*}
 \Expectation \exp(\generalReal \cdot \informationDensity{\inputDistribution}{\effectiveChannelEve(\analogMessageAlphabetElement{1}, \dots, \analogMessageAlphabetElement{\aliceNum}, \cdot, \cdot)}{\inputRV}{\eveOutputRV})
 &=
 \Expectation\left(
   \left(
     \RNDerivative{
       \effectiveChannelEve(\analogMessageAlphabetElement{1}, \dots, \analogMessageAlphabetElement{\aliceNum}, \inputRV, \cdot)
     }{
       \marginalOutputDistribution{\inputDistribution}{\effectiveChannelEve(\analogMessageAlphabetElement{1}, \dots, \analogMessageAlphabetElement{\aliceNum}, \cdot, \cdot)}
     }(\eveOutputRV)
   \right)^\generalReal
 \right)
 \\
 &=
 \exp\left(
  \generalReal
  \cdot
  \frac{1}{\generalReal}
  \log
  \Expectation\left(
    \left(
      \RNDerivative{
        \effectiveChannelEve(\analogMessageAlphabetElement{1}, \dots, \analogMessageAlphabetElement{\aliceNum}, \inputRV, \cdot)
      }{
        \marginalOutputDistribution{\inputDistribution}{\effectiveChannelEve(\analogMessageAlphabetElement{1}, \dots, \analogMessageAlphabetElement{\aliceNum}, \cdot, \cdot)}
      }(\eveOutputRV)
    \right)^\generalReal
  \right)
 \right)
 \\
 &=
 \exp\left(
  \generalReal
  \renyidiv{\generalReal+1}{\inputOutputDistribution{\inputDistribution}{\effectiveChannelEve(\analogMessageAlphabetElement{1}, \dots, \analogMessageAlphabetElement{\aliceNum}, \cdot, \cdot)}}{\inputDistribution\marginalOutputDistribution{\inputDistribution}{\effectiveChannelEve(\analogMessageAlphabetElement{1}, \dots, \analogMessageAlphabetElement{\aliceNum}, \cdot, \cdot)}}
 \right)
\end{align*}
The Rényi divergence appearing at the end is between two multivariate Gaussian distributions and can be seen to be finite from the expressions given in~\cite{gil2011renyi}. In order to verify (\ref{eq:jammersystem-rateconstraint}), we first note that the information expressions appearing are the capacities of the effective channels $\effectiveChannelBob$ and $\effectiveChannelEve$. Since $\stateSpaceElement_1, \dots, \stateSpaceElement_\aliceNum$ change the mean of the channel only, they do not influence the capacity. Therefore, the infimum and supremum are over singleton sets. Consequently, the condition $\fadingJammerBob/\stddevBob > \fadingJammerEve/\stddevEve$ ensures that there is some $\codebookRate$ satisfying (\ref{eq:jammersystem-rateconstraint}).

Fix $\finalconstOne', \finalconstThree'$ as claimed to exist in Theorem~\ref{theorem:jammersystem}, and also fix $\finalconstOne, \finalconstTwo$ with $0<\finalconstTwo<\finalconstOne'$ and $0<\finalconstOne<\finalconstThree'$.

Note that in the \gls{awgn} channel, $\stateSpaceElement_1, \dots, \stateSpaceElement_\aliceNum$ correspond to a shift of the output distribution of the channel, and therefore, the variational distance that appears in (\ref{eq:jammersystem-totvar}) is independent of $\stateSpaceElement_1, \dots, \stateSpaceElement_\aliceNum$. For sufficiently large $\blocklength$, we can therefore fix a codebook $\codebook$ from the $(\inputDistribution, \blocklength, \codebookRate)$-ensemble such that for all $\stateSpaceElement_1, \dots, \stateSpaceElement_\aliceNum$, neither one of the error events described in (\ref{eq:jammersystem-totvar}) and (\ref{eq:jammersystem-reconstruction}) occurs.

Let the jamming strategy be induced by $\codebook_{\costConstraint,\costFunction}$ and let $\eveDecoder: \eveAlphabet^\blocklength \rightarrow [-1,1]$ be an estimator for $\eve$. We bound the \gls{mse} of $\eveDecoder$ as
\begin{align*}
\Expectation_{\eveMarginalOutputDistribution{\effectiveChannelEve^\blocklength(\analogMessageAlphabetElement{1}, \dots, \analogMessageAlphabetElement{\aliceNum}, \cdot, \cdot), \codebook_{\costConstraint,\costFunction}}}
  \left(
    \left(
      \eveDecoder(\eveOutputRV^\blocklength)
      -
      \objectiveFunction(\stateSpaceElement_1, \dots, \stateSpaceElement_\aliceNum)
    \right)^2
  \right)
&=
\int_0^\infty
  \eveMarginalOutputDistribution{\effectiveChannelEve^\blocklength(\analogMessageAlphabetElement{1}, \dots, \analogMessageAlphabetElement{\aliceNum}, \cdot, \cdot), \codebook_{\costConstraint,\costFunction}}
  \left(
    \left(
        \eveDecoder(\eveOutputRV^\blocklength)
        -
        \objectiveFunction(\stateSpaceElement_1, \dots, \stateSpaceElement_\aliceNum)
      \right)^2
    >
    \generalReal
  \right)
d \generalReal
\\
\overset{(a)}&{=}
\int_0^4
  \eveMarginalOutputDistribution{\effectiveChannelEve^\blocklength(\analogMessageAlphabetElement{1}, \dots, \analogMessageAlphabetElement{\aliceNum}, \cdot, \cdot), \codebook_{\costConstraint,\costFunction}}
  \left(
    \left(
        \eveDecoder(\eveOutputRV^\blocklength)
        -
        \objectiveFunction(\stateSpaceElement_1, \dots, \stateSpaceElement_\aliceNum)
      \right)^2
    >
    \generalReal
  \right)
d \generalReal
\\
\overset{(\ref{eq:jammersystem-totvar})}&{\geq}
\int_0^4
  \bigg(
    \marginalOutputDistribution{\inputDistribution}{\effectiveChannelEve(\analogMessageAlphabetElement{1}, \dots, \analogMessageAlphabetElement{\aliceNum}, \cdot, \cdot)}^\blocklength
    \left(
      \left(
          \eveDecoder(\eveOutputRV^\blocklength)
          -
          \objectiveFunction(\stateSpaceElement_1, \dots, \stateSpaceElement_\aliceNum)
        \right)^2
      >
      \generalReal
    \right)
    -
    \exp(-\blocklength\finalconstOne')
  \bigg)
d \generalReal
\\
&=
\Expectation_{\marginalOutputDistribution{\inputDistribution}{\effectiveChannelEve(\analogMessageAlphabetElement{1}, \dots, \analogMessageAlphabetElement{\aliceNum}, \cdot, \cdot)}^\blocklength}
  \left(
    \left(
      \eveDecoder(\eveOutputRV^\blocklength)
      -
      \objectiveFunction(\stateSpaceElement_1, \dots, \stateSpaceElement_\aliceNum)
    \right)^2
  \right)
-4\exp(-\blocklength\finalconstOne').
\end{align*}
where step (a) is due to the fact that both $\eveDecoder(\eveOutputRV^\blocklength)$ and $\objectiveFunction(\stateSpaceElement_1, \dots, \stateSpaceElement_\aliceNum)$ are restricted to the interval [-1,1]. Taking the lower bound for the \gls{mse} under $\marginalOutputDistribution{\inputDistribution}{\effectiveChannelEve(\analogMessageAlphabetElement{1}, \dots, \analogMessageAlphabetElement{\aliceNum}, \cdot, \cdot)}^\blocklength$ from Lemma~\ref{lemma:awgn-scheme-ideal} and noting $\finalconstTwo<\finalconstOne'$, we arrive at the expression in (\ref{eq:awgn-scheme-real-eve}) for sufficiently large $\blocklength$.

For the reconstruction strategy at $\bob$, we first let $\bob$ reconstruct the jamming signal as is possible by Theorem~\ref{theorem:jammersystem} and then post-process the received signal as is possible with knowledge of the jamming signal by Lemma~\ref{lemma:awgn-scheme-ideal}. Using the error bound in Lemma~\ref{lemma:awgn-scheme-ideal} and observing that the maximum instantaneous square error is $4$ since we are constrained to an interval of length $2$ and that $\finalconstOne<\finalconstThree'$, for sufficiently large $\blocklength$ we arrive at (\ref{eq:awgn-scheme-real-bob}).
\end{proof}

\section{Conclusion}
\label{sec:conclusion}
In this work, we have introduced a framework for \acrfull{dfaj}. We have shown how well-known information theoretic tools can be used to improve security by means of a jammer whose signal is stronger at the legitimate receiver than it is at the eavesdropper. In the process, we have proved a compound channel coding result which is a generalization of similar results from the literature.

This work is intended as an initial step towards providing security against eavesdropping for \gls{ota} computation schemes. Our theoretical analysis derives \gls{mse} guarantees both for the eavesdropper's and the legitimate receiver's reconstruction of the objective function for the case in which an arithmetic average is computed over an \gls{awgn} channel. However, a gap between this theoretical work and its implementation for the envisioned practical applications remains. In particular, we are interested in the following questions for future research:
\begin{itemize}
 \item Can the secrecy guarantees in this work be achieved with structured codes which allow for practically feasible encoding and decoding?
 \item Can the secrecy guarantees be strengthened to full semantic security?
 \item Can the approach be generalized to a larger class of channels?
\end{itemize}

\appendix
In this appendix, we prove the two lemmas used for the proof of Lemma~\ref{lemma:awgn-scheme-ideal}.

\begin{proof}[Proof of Lemma~\ref{lemma:bayes-estimator}]
It is known~\cite[eq. (6.92)]{jaynes2003probability} that the \gls{mse} is minimized by the mean of the posterior probability distribution. We can therefore calculate the minimum \gls{mse} estimator given the observations $\generalrvTwoValue_1, \dots, \generalrvTwoValue_\blocklength$ as follows, where we use $\density{}$ with random variables in the index to denote (conditional) densities.
\begin{align*}
\hat{\generalrvOne}
&=
\int_\intervalLowerBound^\intervalUpperBound
  \generalrvOneValue
  \density{\generalrvOne|\generalrvTwo_1, \dots, \generalrvTwo_\blocklength}(\generalrvOneValue | \generalrvTwoValue_1, \dots, \generalrvTwoValue_\blocklength)
d \generalrvOneValue
\\
&\stackrel{(a)}{=}
\int_\intervalLowerBound^\intervalUpperBound
   \generalrvOneValue
   \frac{\density{\generalrvTwo_1, \dots, \generalrvTwo_\blocklength | \generalrvOne}(\generalrvTwoValue, \dots, \generalrvTwoValue_\blocklength | \generalrvOneValue )\density{\generalrvOne}(\generalrvOneValue)}{\density{\generalrvTwo_1, \dots, \generalrvTwo_\blocklength}(\generalrvTwoValue_1, \dots, \generalrvTwoValue_\blocklength)}
d \generalrvOneValue
\\
&=
\frac{
       \int_\intervalLowerBound^\intervalUpperBound
         \generalrvOneValue
         \density{\generalrvTwo_1, \dots, \generalrvTwo_\blocklength | \generalrvOne}(\generalrvTwoValue_1, \dots, \generalrvTwoValue_\blocklength | \generalrvOneValue )\density{\generalrvOne}(\generalrvOneValue)
       d \generalrvOneValue
     }
     {
       \int_\intervalLowerBound^\intervalUpperBound
         \density{\generalrvTwo_1, \dots, \generalrvTwo_\blocklength | \generalrvOne}(\generalrvTwoValue_1, \dots, \generalrvTwoValue_\blocklength | \generalrvOneValue )\density{\generalrvOne}(\generalrvOneValue)
       d \generalrvOneValue
     }
\\
&\stackrel{(b)}{=}
\frac{
       \int_\intervalLowerBound^\intervalUpperBound
         \generalrvOneValue
         \exp\left(
           -\frac{1}{2\stddev^2}
           \sum_{\blockIndex=1}^\blocklength
             (\generalrvTwoValue_\blockIndex - \generalrvOneValue)^2
         \right)
       d \generalrvOneValue
     }
     {
       \int_\intervalLowerBound^\intervalUpperBound
         \exp\left(
           -\frac{1}{2\stddev^2}
           \sum_{\blockIndex=1}^\blocklength
             (\generalrvTwoValue_\blockIndex - \generalrvOneValue)^2
         \right)
       d \generalrvOneValue
     }
\\
&=
\frac{
       \int_\intervalLowerBound^\intervalUpperBound
         \generalrvOneValue
         \exp\left(
           -\frac{1}{2\stddev^2/\blocklength}\left(
             \frac{1}{\blocklength}
             \sum_{\blockIndex=1}^\blocklength
               \generalrvTwoValue_\blockIndex^2
             -
             2 \generalrvOneValue \bar{\generalrvTwoValue}
             +
             \generalrvOneValue^2
           \right)
         \right)
       d \generalrvOneValue
     }
     {
       \int_\intervalLowerBound^\intervalUpperBound
         \exp\left(
           -\frac{1}{2\stddev^2/\blocklength}\left(
             \frac{1}{\blocklength}
             \sum_{\blockIndex=1}^\blocklength
               \generalrvTwoValue_\blockIndex^2
             -
             2 \generalrvOneValue \bar{\generalrvTwoValue}
             +
             \generalrvOneValue^2
           \right)
         \right)
       d \generalrvOneValue
     }
\\
&\stackrel{(c)}{=}
\frac{
       \int_\intervalLowerBound^\intervalUpperBound
         \generalrvOneValue
         \exp\left(
           -\frac{1}{2\stddev^2/\blocklength}\left(
             \bar{\generalrvTwoValue}
             -
             \generalrvOneValue
           \right)^2
         \right)
       d \generalrvOneValue
     }
     {
       \int_\intervalLowerBound^\intervalUpperBound
         \exp\left(
           -\frac{1}{2\stddev^2/\blocklength}\left(
             \bar{\generalrvTwoValue}
             -
             \generalrvOneValue
           \right)^2
         \right)
       d \generalrvOneValue
     }
\end{align*}
For (a), we have applied Bayes' rule. (b) is by observing that $\density{\generalrvOne}(\generalrvOneValue) = 1/(\intervalUpperBound-\intervalLowerBound)$ is independent of $\generalrvOneValue$ in $[\intervalLowerBound,\intervalUpperBound]$ and $\density{\generalrvTwo_1, \dots, \generalrvTwo_\blocklength | \generalrvOne}$ is the normal density. (c) is by multiplying
\[
\exp\left(
         -\frac{1}{2\stddev^2/\blocklength}\left(
           \bar{\generalrvTwoValue}^2
           -
           \frac{1}{\blocklength}
             \sum_{\blockIndex=1}^\blocklength
               \generalrvTwoValue_\blockIndex^2
         \right)
       \right)
\]
on both sides of the fraction to complete the binomials.

The term we have calculated for $\hat{\generalrvOne}$ is the mean of a normal distribution centered at $\bar{\generalrvTwoValue}$ with variance $\stddev^2/\blocklength$ truncated in $[\intervalLowerBound,\intervalUpperBound]$. This is a distribution with a known mean~\cite[eq. 13.134]{johnson1994continuous}, and hence we arrive at (\ref{eq:bayes-estimator}).
\end{proof}

\begin{proof}[Proof of Lemma~\ref{lemma:bayes-estimator-mse}]
Based on the representation (\ref{eq:bayes-estimator}), we calculate the \gls{mse} as follows. We use the substitution rule, substituting $\generalrvTwoValue' := \frac{\bar{\generalrvTwoValue}-\intervalLowerBound}{\stddev/\sqrt{\blocklength}}$ in (a) and $\generalrvOneValue' := \frac{\generalrvOneValue-\intervalLowerBound}{\stddev/\sqrt{\blocklength}}$ in (b).

\begin{align*}
\Expectation\left(
  \left(
    \generalrvOne
    -
    \hat{\generalrvOne}
  \right)^2
\right)
&=
\hspace{-5pt}
\int_\intervalLowerBound^\intervalUpperBound
  \int_{-\infty}^\infty
    \left(
      \bar{\generalrvTwoValue}
      +
      \frac{\stddev}{\sqrt{\blocklength}}
      \cdot
      \frac{\stdnormalpdf\left(\frac{\intervalLowerBound-\bar{\generalrvTwoValue}}{\stddev/\sqrt{\blocklength}}\right)
            -
            \stdnormalpdf\left(\frac{\intervalUpperBound-\bar{\generalrvTwoValue}}{\stddev/\sqrt{\blocklength}}\right)}
          {\stdnormalcdf\left(\frac{\intervalUpperBound-\bar{\generalrvTwoValue}}{\stddev/\sqrt{\blocklength}}\right)
            -
            \stdnormalcdf\left(\frac{\intervalLowerBound-\bar{\generalrvTwoValue}}{\stddev/\sqrt{\blocklength}}\right)}
      -\generalrvOneValue
    \right)^2
\cdot
\frac{1}{\intervalUpperBound-\intervalLowerBound}
\cdot
\frac{1}{\stddev/\sqrt{\blocklength}}
\stdnormalpdf\left( \frac{\generalrvOneValue-\bar{\generalrvTwoValue}}{\stddev/\sqrt{\blocklength}} \right)
d \bar{\generalrvTwoValue}
d \generalrvOneValue
\\
&\stackrel{(a)}{=}
\hspace{-7pt}
\int_\intervalLowerBound^\intervalUpperBound
  \int_{-\infty}^\infty
    \left(\rule{0cm}{1cm}\right.
      \frac{\stddev}{\sqrt{\blocklength}}
      \left(
        \generalrvTwoValue'
        +
        \frac{\stdnormalpdf(-\generalrvTwoValue')
              -
              \stdnormalpdf\left(\frac{\intervalUpperBound-\intervalLowerBound}{\stddev/\sqrt{\blocklength}}-\generalrvTwoValue' \right)}
            {\stdnormalcdf\left(\frac{\intervalUpperBound-\intervalLowerBound}{\stddev/\sqrt{\blocklength}}-\generalrvTwoValue' \right)
              -
              \stdnormalcdf(-\generalrvTwoValue')}
      \right)
      +\intervalLowerBound
      -\generalrvOneValue
    \left.\rule{0cm}{1cm}\right)^2
\cdot
\frac{1}{\intervalUpperBound-\intervalLowerBound}
\cdot
\stdnormalpdf\left(\frac{\generalrvOneValue-\intervalLowerBound}{\stddev/\sqrt{\blocklength}} - \generalrvTwoValue'\right)
d \generalrvTwoValue'
d \generalrvOneValue
\\
&\stackrel{(b)}{=}
\hspace{-7pt}
\int_0^{\frac{\intervalUpperBound-\intervalLowerBound}{\stddev/\sqrt{\blocklength}}}
  \int_{-\infty}^\infty
      \left(
        \generalrvTwoValue'
        +
        \frac{\stdnormalpdf(-\generalrvTwoValue')
              -
              \stdnormalpdf\left(\frac{\intervalUpperBound-\intervalLowerBound}{\stddev/\sqrt{\blocklength}}-\generalrvTwoValue' \right)}
            {\stdnormalcdf\left(\frac{\intervalUpperBound-\intervalLowerBound}{\stddev/\sqrt{\blocklength}}-\generalrvTwoValue' \right)
              -
              \stdnormalcdf(-\generalrvTwoValue')}
        -\generalrvOneValue'
      \right)^2
\cdot
\left(\frac{\stddev}{\sqrt{\blocklength}}\right)^3
\cdot
\frac{1}{\intervalUpperBound-\intervalLowerBound}
\cdot
\stdnormalpdf(\generalrvOneValue' - \generalrvTwoValue')
d \generalrvTwoValue'
d \generalrvOneValue'
\\
&=
\frac{\stddev^2}{\blocklength}
\mseFunction\left(
  \frac{\intervalUpperBound-\intervalLowerBound}{\stddev/\sqrt{\blocklength}}
\right),
\end{align*}
concluding the proof of the lemma.
\end{proof}

\bibliographystyle{IEEEtran}
\bibliography{references}

\begin{thebibliography}{10}
\providecommand{\url}[1]{#1}
\csname url@samestyle\endcsname
\providecommand{\newblock}{\relax}
\providecommand{\bibinfo}[2]{#2}
\providecommand{\BIBentrySTDinterwordspacing}{\spaceskip=0pt\relax}
\providecommand{\BIBentryALTinterwordstretchfactor}{4}
\providecommand{\BIBentryALTinterwordspacing}{\spaceskip=\fontdimen2\font plus
\BIBentryALTinterwordstretchfactor\fontdimen3\font minus
  \fontdimen4\font\relax}
\providecommand{\BIBforeignlanguage}[2]{{%
\expandafter\ifx\csname l@#1\endcsname\relax
\typeout{** WARNING: IEEEtran.bst: No hyphenation pattern has been}%
\typeout{** loaded for the language `#1'. Using the pattern for}%
\typeout{** the default language instead.}%
\else
\language=\csname l@#1\endcsname
\fi
#2}}
\providecommand{\BIBdecl}{\relax}
\BIBdecl

\bibitem{amiri2020machine}
M.~M. Amiri and D.~G{\"u}nd{\"u}z, ``Machine learning at the wireless edge:
  Distributed stochastic gradient descent over-the-air,'' \emph{IEEE
  Transactions on Signal Processing}, vol.~68, pp. 2155--2169, 2020.

\bibitem{kiril}
K.~Ralinovski, M.~Goldenbaum, and S.~Sta{\'n}czak, ``Energy-efficient
  classification for anomaly detection: The wireless channel as a helper,'' in
  \emph{2016 IEEE International Conference on Communications (ICC)}, 2016, pp.
  1--6.

\bibitem{gastpar2003source}
M.~Gastpar and M.~Vetterli, ``Source-channel communication in sensor
  networks,'' in \emph{Information Processing in Sensor Networks}, F.~Zhao and
  L.~Guibas, Eds.\hskip 1em plus 0.5em minus 0.4em\relax Berlin and Heidelberg,
  Germany: Springer, 2003, pp. 162--177.

\bibitem{goldenbaum2013robust}
M.~Goldenbaum and S.~Stanczak, ``Robust analog function computation via
  wireless multiple-access channels,'' \emph{IEEE Transactions on
  Communications}, vol.~61, no.~9, pp. 3863--3877, 2013.

\bibitem{bjelakovic2019distributed}
I.~Bjelakovi{\'c}, M.~Frey, and S.~Sta{\'n}czak, ``Distributed approximation of
  functions over fast fading channels with applications to distributed learning
  and the max-consensus problem,'' in \emph{2019 57th Annual Allerton
  Conference on Communication, Control, and Computing}.\hskip 1em plus 0.5em
  minus 0.4em\relax IEEE, 2019, pp. 1146--1153.

\bibitem{liu2020over}
W.~Liu, X.~Zang, Y.~Li, and B.~Vucetic, ``Over-the-air computation systems:
  Optimization, analysis and scaling laws,'' \emph{IEEE Transactions on
  Wireless Communications}, vol.~19, pp. 5488--5502, 2020.

\bibitem{frey2021over}
M.~Frey, I.~Bjelaković, and S.~Stańczak, ``Over-the-air computation in
  correlated channels,'' \emph{IEEE Transactions on Signal Processing},
  vol.~69, pp. 5739--5755, 2021.

\bibitem{nazer2007computation}
B.~Nazer and M.~Gastpar, ``Computation over multiple-access channels,''
  \emph{IEEE Transactions on Information Theory}, vol.~53, no.~10, pp.
  3498--3516, 2007.

\bibitem{dobrushin1959optimum}
R.~L. Dobrushin, ``Optimum information transmission through a channel with
  unknown parameters,'' \emph{Radio Engineering and Electronics}, vol.~4,
  no.~12, pp. 1--8, 1959.

\bibitem{blackwell1959capacity}
D.~Blackwell, L.~Breiman, and A.~Thomasian, ``The capacity of a class of
  channels,'' \emph{The Annals of Mathematical Statistics}, pp. 1229--1241,
  1959.

\bibitem{wolfowitz1959simultaneous}
J.~Wolfowitz, ``Simultaneous channels,'' \emph{Archive for Rational Mechanics
  and Analysis}, vol.~4, pp. 371--386, 1959.

\bibitem{kesten1961some}
H.~Kesten, ``Some remarks on the capacity of compound channels in the
  semicontinuous case,'' \emph{Information and Control}, vol.~4, no. 2-3, pp.
  169--184, 1961.

\bibitem{yoshihara1965coding}
K.~Yoshihara, ``Coding theorems for the compound semi-continuous memoryless
  channels,'' in \emph{Kodai Mathematical Seminar Reports}, vol.~17,
  no.~1.\hskip 1em plus 0.5em minus 0.4em\relax Department of Mathematics,
  Tokyo Institute of Technology, 1965, pp. 30--43.

\bibitem{ahlswede1967certain}
R.~Ahlswede, ``Certain results in coding theory for compound channels,'' in
  \emph{Proceedings of the Colloquium on Information Theory Debrecen
  (Hungary)}, 1967, pp. 35--60.

\bibitem{root1968capacity}
W.~Root and P.~Varaiya, ``Capacity of classes of {G}aussian channels,''
  \emph{SIAM Journal on Applied Mathematics}, vol.~16, no.~6, pp. 1350--1393,
  1968.

\bibitem{bjelakovic2013secrecy}
I.~Bjelakovi{\'c}, H.~Boche, and J.~Sommerfeld, ``Secrecy results for compound
  wiretap channels,'' \emph{Problems of Information Transmission}, vol.~49,
  no.~1, pp. 73--98, 2013.

\bibitem{he2014mimo}
X.~He and A.~Yener, ``{MIMO} wiretap channels with unknown and varying
  eavesdropper channel states,'' \emph{IEEE Transactions on Information
  Theory}, vol.~60, no.~11, pp. 6844--6869, 2014.

\bibitem{wyner1975common}
A.~Wyner, ``The common information of two dependent random variables,''
  \emph{IEEE Transactions on Information Theory}, vol.~21, no.~2, pp. 163--179,
  1975.

\bibitem{han1993approximation}
T.~S. Han and S.~Verd{\'u}, ``Approximation theory of output statistics,''
  \emph{IEEE Transactions on Information Theory}, vol.~39, no.~3, pp. 752--772,
  1993.

\bibitem{csiszar1996almost}
I.~Csiszár, ``Almost independence and secrecy capacity,'' \emph{Problems of
  Information Transmission}, vol.~32, no.~1, pp. 40--47, 1996.

\bibitem{devetak2005private}
I.~Devetak, ``The private classical capacity and quantum capacity of a quantum
  channel,'' \emph{IEEE Transactions on Information Theory}, vol.~51, no.~1,
  pp. 44--55, 2005.

\bibitem{hayashi2016secure}
M.~Hayashi and R.~Matsumoto, ``Secure multiplex coding with dependent and
  non-uniform multiple messages,'' \emph{IEEE Transactions on Information
  Theory}, vol.~62, no.~5, pp. 2355--2409, 2016.

\bibitem{cuff2016soft}
P.~Cuff, ``Soft covering with high probability,'' in \emph{2016 IEEE
  International Symposium on Information Theory}.\hskip 1em plus 0.5em minus
  0.4em\relax IEEE, 2016, pp. 2963--2967.

\bibitem{frey2018resolvability}
M.~Frey, I.~Bjelakovic, and S.~Stanczak, ``Resolvability on continuous
  alphabets,'' in \emph{2018 IEEE International Symposium on Information
  Theory}.\hskip 1em plus 0.5em minus 0.4em\relax IEEE, 2018, pp. 2037--2041.

\bibitem{hayashi2006general}
M.~Hayashi, ``General nonasymptotic and asymptotic formulas in channel
  resolvability and identification capacity and their application to the
  wiretap channel,'' \emph{IEEE Transactions on Information Theory}, vol.~52,
  no.~4, pp. 1562--1575, 2006.

\bibitem{csiszar2011information}
I.~Csiszar and J.~K{\"o}rner, \emph{Information theory: Coding theorems for
  discrete memoryless systems}.\hskip 1em plus 0.5em minus 0.4em\relax
  Cambridge: Cambridge University Press, 2011.

\bibitem{BellareSemantic}
M.~Bellare, S.~Tessaro, and A.~Vardy, ``Semantic security for the wiretap
  channel,'' in \emph{Advances in Cryptology--CRYPTO 2012}.\hskip 1em plus
  0.5em minus 0.4em\relax Springer, 2012, pp. 294--311.

\bibitem{BlochStrongSecrecy}
M.~R. Bloch and J.~N. Laneman, ``Strong secrecy from channel resolvability,''
  \emph{IEEE Transactions on Information Theory}, vol.~59, no.~12, pp.
  8077--8098, 2013.

\bibitem{negi2005secret}
R.~Negi and S.~Goel, ``Secret communication using artificial noise,'' in
  \emph{IEEE Vehicular Technology Conference}, vol.~62, no.~3.\hskip 1em plus
  0.5em minus 0.4em\relax IEEE, 2005, p. 1906.

\bibitem{vilela2010friendly}
J.~P. Vilela, M.~Bloch, J.~Barros, and S.~W. McLaughlin, ``Friendly jamming for
  wireless secrecy,'' in \emph{2010 IEEE International Conference on
  Communications}.\hskip 1em plus 0.5em minus 0.4em\relax IEEE, 2010, pp. 1--6.

\bibitem{vilela2011wireless}
------, ``Wireless secrecy regions with friendly jamming,'' \emph{IEEE
  Transactions on Information Forensics \& Security}, vol.~6, no.~2, pp.
  256--266, 2011.

\bibitem{stanojev2012improving}
I.~Stanojev and A.~Yener, ``Improving secrecy rate via spectrum leasing for
  friendly jamming,'' \emph{IEEE Transactions on Wireless Communications},
  vol.~12, no.~1, pp. 134--145, 2012.

\bibitem{tekin2008general}
E.~Tekin and A.~Yener, ``The general {G}aussian multiple-access and two-way
  wiretap channels: Achievable rates and cooperative jamming,'' \emph{IEEE
  Transactions on Information Theory}, vol.~54, no.~6, pp. 2735--2751, 2008.

\bibitem{pierrot2011strongly}
A.~J. Pierrot and M.~R. Bloch, ``Strongly secure communications over the
  two-way wiretap channel,'' \emph{IEEE Transactions on Information Forensics
  and Security}, vol.~6, no.~3, pp. 595--605, 2011.

\bibitem{he2013role}
X.~He and A.~Yener, ``The role of feedback in two-way secure communications,''
  \emph{IEEE Transactions on Information Theory}, vol.~59, no.~12, pp.
  8115--8130, 2013.

\bibitem{ShannonSecrecy}
C.~E. Shannon, ``Communication theory of secrecy systems,'' \emph{Bell Labs
  Technical Journal}, vol.~28, no.~4, pp. 656--715, 1949.

\bibitem{WynerWiretap}
A.~D. Wyner, ``The wire-tap channel,'' \emph{Bell Labs Technical Journal},
  vol.~54, no.~8, pp. 1355--1387, 1975.

\bibitem{MaurerStrongSecret}
U.~M. Maurer, ``The strong secret key rate of discrete random triples,'' in
  \emph{Communications and Cryptography: Two Sides of One Tapestry}, R.~E.
  Blahut, D.~J. Costello, U.~Maurer, and T.~Mittelholzer, Eds.\hskip 1em plus
  0.5em minus 0.4em\relax Boston: Springer, 1994, pp. 271--285.

\bibitem{HouEffectiveSecrecy}
J.~Hou and G.~Kramer, ``Effective secrecy: Reliability, confusion and
  stealth,'' in \emph{2014 IEEE International Symposium on Information
  Theory}.\hskip 1em plus 0.5em minus 0.4em\relax IEEE, 2014, pp. 601--605.

\bibitem{doliveira2018computational}
R.~G. D‘Oliveira, S.~El~Rouayheb, and M.~M{\'e}dard, ``The computational
  wiretap channel,'' in \emph{2018 56th Annual Allerton Conference on
  Communication, Control, and Computing}.\hskip 1em plus 0.5em minus
  0.4em\relax IEEE, 2018, pp. 1136--1140.

\bibitem{bassi2019mutual}
G.~Bassi and M.~Skoglund, ``On the mutual information of two boolean functions,
  with application to privacy,'' in \emph{2019 IEEE International Symposium on
  Information Theory}.\hskip 1em plus 0.5em minus 0.4em\relax IEEE, 2019, pp.
  1197--1201.

\bibitem{bertsekas2008introduction}
D.~Bertsekas and J.~N. Tsitsiklis, \emph{Introduction to Probability},
  2nd~ed.\hskip 1em plus 0.5em minus 0.4em\relax Belmont: Athena Scientific,
  2008.

\bibitem{utkovski2019learning}
Z.~Utkovski, P.~Agostini, M.~Frey, I.~Bjelakovic, and S.~Stanczak, ``Learning
  radio maps for physical-layer security in the radio access,'' in \emph{2019
  IEEE 20th International Workshop on Signal Processing Advances in Wireless
  Communications}.\hskip 1em plus 0.5em minus 0.4em\relax IEEE, 2019.

\bibitem{klinc2011ldpc}
D.~Klinc, J.~Ha, S.~W. McLaughlin, J.~Barros, and B.-J. Kwak, ``{LDPC} codes
  for the gaussian wiretap channel,'' \emph{IEEE Transactions on Information
  Forensics \& Security}, vol.~6, no.~3, pp. 532--540, 2011.

\bibitem{vanerven2014renyi}
T.~van Erven and P.~Harremos, ``{R}{\'e}nyi divergence and {K}ullback-{L}eibler
  divergence,'' \emph{IEEE Transactions on Information Theory}, vol.~60, no.~7,
  pp. 3797--3820, 2014.

\bibitem{gil2011renyi}
M.~Gil, ``On {R}ényi divergence measures for continuous alphabet sources,''
  Master's thesis, Queen's University Kingston, Ontario, Canada, 2011.

\bibitem{shannon1948mathematical}
C.~E. Shannon, ``A mathematical theory of communication,'' \emph{The Bell
  System Technical Journal}, vol.~27, no.~3, pp. 379--423, 1948.

\bibitem{elgamal2011network}
A.~El~Gamal and Y.-H. Kim, \emph{Network Information Theory}.\hskip 1em plus
  0.5em minus 0.4em\relax Cambridge: Cambridge University Press, 2011.

\bibitem{ahlswede1978elimination}
R.~Ahlswede, ``Elimination of correlation in random codes for arbitrarily
  varying channels,'' \emph{Zeitschrift f{\"u}r Wahrscheinlichkeitstheorie und
  verwandte Gebiete}, vol.~44, no.~2, 1978.

\bibitem{dubhashi2009concentration}
D.~P. Dubhashi and A.~Panconesi, \emph{Concentration of Measure for the
  Analysis of Randomized Algorithms}.\hskip 1em plus 0.5em minus 0.4em\relax
  Cambridge: Cambridge University Press, 2009.

\bibitem{jaynes2003probability}
E.~Jaynes, \emph{Probability Theory: The Logic of Science}.\hskip 1em plus
  0.5em minus 0.4em\relax Cambridge: Cambridge University Press, 2003.

\bibitem{johnson1994continuous}
N.~L. Johnson, S.~Kotz, and N.~Balakrishnan, \emph{Continuous Univariate
  Distributions}, 2nd~ed., ser. Wiley Series in Probability and Mathematical
  Statistics.\hskip 1em plus 0.5em minus 0.4em\relax New York, Chichester,
  Brisbane, Toronto, Singapore: Wiley, 1994, vol.~1.

\end{thebibliography}

\end{document}